\documentclass[11pt]{article}
\usepackage[a4paper, top=2.5cm, bottom=2.5cm, left=2.5cm, right=2.5cm]
{geometry}

\usepackage{graphicx}
\usepackage{amsmath}
\usepackage{amsthm}
\usepackage{amsfonts}
\usepackage{amssymb}
\usepackage{natbib}
\usepackage{makecell}
\usepackage{algorithmic}
\usepackage[linesnumbered,algoruled,boxed,lined]{algorithm2e}
\usepackage{xcolor}
\usepackage{subcaption}
\usepackage{soul}
\usepackage{multirow}
\usepackage{hyperref}
\usepackage{booktabs}

\newtheorem{definition}{Definition}
\newtheorem{proposition}{Proposition}
\newtheorem{corollary}{Corollary}
\newtheorem{theorem}{Theorem}

\newtheorem{lemma}{Lemma}

\title{Achieving Equitability with Subsidy\footnote{This is the second version. We have fixed some errors and updated the related work.}}

\author{
Yuanyuan Wang$^1$ \hspace{15pt}
Tianze Wei$^2$ \hspace{15pt}\\ 
$^1$School of Mathematical Science, The Ocean University of China \\ \small
\texttt{wyy8088@stu.ouc.edu.cn} \\ 
$^2$Department of Computer Science, City University of Hong Kong \\ \small
\texttt{t.z.wei-8@my.cityu.edu.hk} \\ 
}

\date{}

\begin{document}
\maketitle

\begin{abstract}
We study the fair allocation problem of indivisible items with subsidy.
In this paper, we focus on the notion of fairness - equitability (EQ), which requires that items be allocated such that all agents value the bundle they receive equally.
First, we study the upper bounds of the minimum required subsidy to achieve EQ in different item settings and provide the corresponding lower bounds.
Second, we consider the bounded subsidy for achieving EQ and another popular notion of fairness - envy-freeness (EF), and give a characterization of allocations that can achieve both EQ and EF.
Finally, we analyze the bounds of subsidy of allocations achieving fairness and efficiency (utilitarian social welfare or Nash welfare) and design several polynomial-time algorithms to compute the desired allocation.
\end{abstract}

\section{Introduction}
The fair division involves \textit{fairly} distributing a set of items among a group of agents with heterogeneous preferences. 
This problem has received significant attention in diverse research communities, such as algorithmic game theory and artificial intelligence, and its relevant theory already has practical applications like the widely-used fair division platform spliddit.org \cite{goldman2015spliddit}.
The notion of fairness can be conceptualized from various perspectives, leading to different formal definitions.
Two prominent fairness criteria are \textit{envy-freeness} (EF), where each agent prefers her own bundle of items to any other agent's bundle, and \textit{equitability} (EQ), which requires that all agents receive bundles of equal value. 
In our paper, we focus on the fairness criterion - EQ, whose practical importance has been empirically validated through experimental studies, such as the free-form bargaining experiments \cite{herreiner2009envy}.
Intuitively, EQ aims to minimize the value gap between the highest and lowest valued bundles, and notably coincides with EF when agents have identical valuations.


For divisible items (e.g., cakes), EQ allocations are guaranteed to exist by the ``Moving Knife Procedure'' \cite{segal2018resource}. 
However, in the setting of indivisible items, where each item must be allocated integrally to an agent, EQ allocations cannot be guaranteed. 
Consider a simple example with two agents and a single piece of jewelry that both agents highly value.
It is trivial that no allocation can achieve EQ.
Consequently, much of the recent research has focused on the relaxation of equitability, such as equitability up to one item (EQ1) \cite{freeman2019equitable,hosseini2025equitable} and equitability up to any item (EQX) \cite{gourves2014near,freeman2019equitable, barman2024nearly}, 
in order to approximate exact equitability (EQ) as closely as possible.
An alternative approach to achieving fairness involves monetary compensation.
Returning to our jewelry example. 
EQ can be achieved by allocating agent 1 this jewelry while providing agent 2 with the money compensation equal to agent 1's valuation of the jewelry.
In this paper, we study whether EQ could be achieved by the addition of a small quantity of money, namely \textit{subsidy}, and further proceed to consider whether both EQ and EF can be guaranteed with the same limited subsidy.
There exists a substantial body of work on achieving precise fairness through subsidy, with detailed discussions available in the relevant literature \cite{halpern2019fair,brustle2020one,barman2022achieving,goko2024fair,wu2023one,wu2024tree}.

Beyond fairness considerations, the fair division also focuses on efficiency that maximizes the collective welfare of agents. 
In our work, we focus on utilitarian social welfare and Nash welfare as primary efficiency measures. 
Then, a natural question arises: can an allocation achieve "optimally bounded" subsidy, where the upper bound of required minimum subsidy matches the lower bound, while simultaneously guaranteeing efficiency?
This motivates our investigation into the compatibility between ``optimally bounded'' subsidy and efficiency.
To summarize, our paper studies the fair allocation of $m$ indivisible items among a set of $n$ agents with subsidy and addresses the following fundamental questions:
\begin{quote}
    \textit{ How much subsidy is required to achieve EQ?
     What condition must be satisfied for an allocation to achieve both EQ and EF with the same subsidy?
     What efficiency can an allocation guarantee when its required minimum subsidy is ``optimally bounded''?}
\end{quote}


\subsection{Our Contribution}
\begin{table}[tb]
    \centering
       \small \setlength{\tabcolsep}{1.5mm}{
    \begin{tabular}{c|c|c|c}
    \toprule
        Items &  Lower Bound
         &Upper Bound&Property\\
    \midrule
    Goods (Chores) & \makecell[c]{$n-1$ (Prop \ref{prop: lower_bound_subsidy_allocation_can_be_chosen})}
    & \makecell[c]{$n-1$ (Thm \ref{the: nonomontone_subsidy})} & EC\\
    \midrule
    \makecell[c]{Mixed Items} & \makecell[c]{$\max \{n-1, m\}$ (Prop \ref{prop: lower_bound_subsidy_allocation_can_be_chosen})} & \makecell[c]{$(n-1)^{\clubsuit}$ or $(n+m-2)^{\spadesuit}$ (Thm \ref{the: nonomontone_subsidy})}&EC\\
     \midrule
    Goods &\makecell[c]{$(n-1)m$ (Prop \ref{prop: subsidy_ef_eq_lower_bound})}&\makecell[c]{$(n-1)m$ (Coro \ref{coro: susbsidy_ef_eq_worst_case})} & EC\&EFB \\
    \midrule
    Chores &\makecell[c]{$m$ (Prop \ref{prop: subsidy_ef_eq_lower_bound})}&\makecell[c]{$
    m$ (Coro \ref{coro: susbsidy_ef_eq_worst_case})} & EC\&EFB \\
    \bottomrule
    \end{tabular}
    }
    \caption{The lower and upper bounds of the required minimum subsidy.
    EC means equitable-convertible (Definition \ref{def: ec}), EFB means envy-freeable (Definition \ref{def: efb}), and EC\&EFB means that both equitable-convertible and envy-freeable (Definition \ref{def: ec_efb}).
    $\clubsuit$ means objective valuations and $\spadesuit$ means additive valuations. }
    \label{tab: subsidy_eq}
\end{table}

In this paper, we study allocating $m$ indivisible items to $n$ agents. 
The first part focuses on the minimum subsidy that the allocation needs for achieving fairness, including EQ and EF, and the main results are summarized in Table \ref{tab: subsidy_eq}. 
More specifically, for EQ exclusively, we establish the following results:

\begin{itemize}
    \item For the goods or chores setting, we show that the lower bound of required minimum subsidy is $(n-1)$. For the mixed items setting, we show that the lower bound of required minimum subsidy is $\max\{n-1,m\}$.
    \item We propose the Squeeze Algorithm (Algorithm \ref{alg: nonmonotone_subsidy}) to compute an allocation that requires the minimum subsidy at most (1) $(n-1)$ to achieve EQ for the goods or chores setting, or (2) $(n-1)$ or $(n+m-2)$ to achieve EQ for the mixed items setting with objective valuations or additive valuations respectively.
\end{itemize}
For the combination of EQ and EF, we obtain the following results:
\begin{itemize}
    \item We provide a characterization for allocations simultaneously achieving EQ and EF with the same payment vector. 
    \item For the goods or chores setting, we show that as long as an allocation is equitable-convertible and envy-freeable with the same payment vector, the required minimum subsidy to achieve both EQ and EF is at most $(n-1)m$ or $m$, which is tight.
\end{itemize}
The second part explores fairness with the ``optimally bounded'' subsidy, maintaining efficiency in the goods setting, and the main results are summarized in Table \ref{tab: subsidy_fairness_efficiency}.  
More specifically, for EQ exclusively, we have the following results:
\begin{itemize}
\item For additive valuations, 
we construct a counterexample demonstrating that allocations achieving EQ with at most $(n-1)$ subsidy cannot guarantee any efficiency. 
    \item For normalized subadditive valuations, we propose the Balanced-Packing Algorithm (Algorithm \ref{alg: eq1_subadditive_sw}) to compute an allocation that is EQ1 and guarantees a tight $O(\frac{1}{n})$-max-USW in polynomial time. For two agents, we present a polynomial time algorithm computing an allocation that is EQ1 and $\frac{1}{3}$-max-NSW.
    Additionally, in both cases, the required minimum subsidy is at most $(n-1)$.
    \item For normalized additive valuations, there exists an allocation that is EQ1 and guarantees a $\frac{1}{\sqrt{2}}$-max-NSW for two agents, which requires at most $(n-1)$ minimum subsidy. 

    \end{itemize}
For EQ and EF jointly, we obtain the following results:
\begin{itemize}
    \item For additive valuations, we prove that the maximum utilitarian social welfare allocation is both equitable-convertible and envy-freeable.
    Additionally, we construct a counterexample to show that any equitable-convertible and envy-freeable allocation cannot guarantee any efficiency in terms of Nash welfare. 
    \item For matroid rank valuations, we show that there exist allocations that maximize Nash welfare or utilitarian social welfare, are both equitable-convertible and envy-freeable, and they can be computed in polynomial time.
\end{itemize}

\begin{table}[tb]
    \centering
    \small
    \setlength{\tabcolsep}{0.06cm}{
    \begin{tabular}{c|c|c|c|c}

    \toprule
       Valuation&\makecell[c]{The number of \\ Agents} & \makecell[c]{Upper Bound \\of Subsidy} &\makecell[c]{Fairness/Efficiency}
         &Property 
         \\         
    \midrule
    Additive & $n$ & $n-1$ & \makecell[c]{Inapproximability of USW/NSW (Prop \ref{the: subsidy_no_efficiency_scaled})} &   EC\\
    \midrule
    \makecell[c]{Subadditive \\ (Normalized)} & $n$ & $n-1$ & \makecell[c]{EQ1+$O(\frac{1}{n})$-Max-USW
    (Thm \ref{the: usw_subadditive})} & EC\\
    \midrule
    \makecell[c]{Subadditive \\ (Normalized)} & 2 & $1$ & \makecell[c]{EQ1+$\frac{1}{3}$-Max-NSW
    (Thm \ref{the: subadditive_two_agents_nash_welfare})} & EC\\
    \midrule
     \makecell[c]{Additive\\ (Normalized)} & 2 & $1$ & \makecell[c]{EQ1+$\frac{1}{\sqrt{2}}$-Max-NSW
    (Thm \ref{the: additive_two_agents_nash_welfare_existence})}& EC\\
    \midrule
    Additive & $n$ & $(n-1)m$ & \makecell[c]{Inapproximability of NSW (Prop \ref{prop: no_efficiency_ef_eq})\\ Max-USW (Prop \ref{the: additive_utilitarian_eq_ef})} &   EC\&EFB\\
    \midrule
    Matroid Rank & $n$ & $(n-1)m$ & \makecell[c]{Max-USW/NSW (Thm \ref{the: matroid_rank_eq_ef})} &   EC\&EFB\\
    \bottomrule
    \end{tabular}
    }
    \caption{Bounded subsidy and efficiency, Max-NSW means maximum Nash welfare, Max-USW means maximum utilitarian social welfare, EC means equitable-convertible, and EFB means envy-freeable.}
    \label{tab: subsidy_fairness_efficiency}
\end{table}

    

\subsection{Related Work}
In the fair allocation of indivisible items, there is a line of work about the relaxation of EQ. 
\cite{gourves2014near} shows that EQX allocation always exists for additive valuations and can be computed in polynomial time.
\cite{barman2024nearly} studies the existence of EQX allocations in the goods setting when agents have beyond additive valuations,
where the existence algorithm is not executed in polynomial time.
\cite{hosseini2025equitable} studies EQ1 allocations in the setting of mixed goods and chores when agents have additive valuations. 
In our paper, we study the exact EQ fairness with the subsidy, where one side result is the polynomial time computation of EQ1 allocations for the goods or chores setting beyond additive valuations.

Moreover, money compensation has been prevalent in classic economic literature \cite{maskin1987fair,alkan1991fair}.
When agents have additive valuations,
\cite{halpern2019fair} 
characterizes 
an allocation that can be 
envy-free with some payments, which is called envy-freeable allocation,  
and show that an envy-freeable allocation requires a subsidy of 
at most $(n-1)m$.
\cite{brustle2020one} propose a polynomial time algorithm to compute an envy-freeable allocation where each agent's subsidy is at most one.
Later, \cite{caragiannis2021computing} studies the optimization of computing the minimum subsidy and showed the additive approximation algorithm for bounded subsidy.
\cite{caragiannis2021interim} initiates the study of interim envy-freeness (iEF) allocation with subsidy and focuses on several different payments.
When agents have monotone valuations,
\cite{brustle2020one} proves that there exists an envy-freeable allocation where the subsidy to each agent is at most $2(n-1)$.
Consequently, \cite{kawase2024towards} improves the upper bound of each agent's subsidy to at most $(n-\frac{3}{2})$.
In the restricted valuation domain, \cite{barman2022achieving} studies the general valuations with binary marginal and showed the optimal upper bound $(n-1)$ of the total subsidy.
\cite{goko2024fair} shows that when agents have matroid valuations, there is a truthful allocation mechanism that achieves envy-freeness and utilitarian optimality by subsidizing each agent with at most one.
Money compensation is also studied in other fairness notions, such as proportional allocations with subsidy \cite{wu2023one, wu2024tree}, EF and EQ allocations with money transfer for agents with superadditive valuations \cite{aziz2021achieving}, where the negative payment is allowed, and the total subsidy is zero, which is different from our setting.
In our paper, we aim to address the gap in the literature on fair allocation with subsidy, where we consider EQ fairness.

For the fair division problem that satisfies both fairness and efficiency, \cite{freeman2019equitable,freeman2020equitable} explore the compatibility between EQX (EQ1) and Pareto optimality,
and
\cite{sun2023equitability} studies the equitability and utilitarian social welfare in allocating indivisible items when agents have additive valuations.
In addition, 
\cite{narayan2021two} shows that for general monotone valuations, there always exists an EF allocation with transfer payment, which guarantees a high fraction of optimal Nash welfare, and they also study the compatibility with utilitarian social welfare.

\section{Preliminaries}
For $k \in \mathbb{N}$, let $[k] = \{1, \ldots, k\}$. 
Let $N = [n]$ denote the set of agents and let $M = \{ e_1, \ldots, e_m\}$ denote the set of $m$ indivisible items.
Each agent $i \in N$ is endowed with a valuation function $v_i: 2^{M} \rightarrow \mathbb{R}$.
We assume that $v_i(\emptyset)=0$.
For any valuation function $v_i$, we define the marginal value of an item $e\in M$ with respect to a bundle $S\subseteq M$,
as $\Delta_i(S,e)=v_i(S \cup e) - v_i(S)$. If $\Delta_i(S,e)\ge 0$ for any $S \subseteq M$, we say that an item $e \in M$ is a \textit{good} for agent $i \in N$.
Additionally, if $\Delta_i(S,e)< 0$ for some $S \subseteq M$, we say that an item $e \in M$ is a \textit{chore} for agent $i \in N$.
We consider three different settings of items: (1) all items are goods, (2) all items are chores, and (3) mixtures of goods and chores, where an agent may derive positive, negative, or zero marginal valuation from an item. 

A valuation function $v_i$ is additive if $v_i(S) =\sum_{e \in S}v_i(e)$ for any $S \subseteq M$, is submodular if $ v_i(S) + v_i(T) \geq v_i(S \cup T) + v_i(S \cap T)$ for any $S, T \subseteq M$, and is subadditive if $v_i(S \cup T) \leq v_i(S) + v_i(T)$ for any $S, T \subseteq M$. 
A valuation $v_i$ is called a matroid rank function if $v_i$ is submodular and the marginal value of each item is zero or one.
W.L.O.G., we assume a standard value oracle access to valuations, that given any agent $i \in N$ and any subset $S \subseteq M$, it returns $v_i (S) \in \mathbb{R}$ in $O(1)$ time.
Unless specified otherwise, we assume that each agent's valuation function $v_i$ is scaled so that the maximum marginal value of each item is at most one, i.e., $\lvert \Delta_i(S,e) \rvert \leq 1$ for any $i \in N$, $e \in M$ and $S \subseteq M \setminus \{e\}$.
For a singleton set $\{e\}$, we will use $v_i(x)$ as a shorthand for $v_i(\{x\})$.
A fair allocation instance is denoted as $\mathcal{I} = \langle M, N, \boldsymbol{v} \rangle$, where $\boldsymbol{v} = (v_1, \ldots, v_n)$.
An allocation ${\bf A} = (A_1, \ldots, A_n)$ is a $n$-partition of $M$, where $A_i$ is the bundle allocated to agent $i$, $\bigcup_{i \in N}A_{i} = M$ and $A_{i} \cap A_{j} = \emptyset$ for any two agents $i \neq j$.
Let $\Pi_{n}(M)$ denote the set of all $n$-partitions of the items. 

\subsection{Fairness and Efficiency Notions}

\begin{definition}[Equitability]

An allocation ${\bf A}$ is equitable (EQ) if for any two agents $i,j \in N$, $v_i(A_i) = v_j(A_j)$.
    
\end{definition}

\begin{definition}
[Equitability up to One Item]
\label{def: eq1}
    
An allocation ${\bf A}$ is equitable up to one item (EQ1), if for any two agents $i,j \in N$ with $A_i, A_j \neq \emptyset$, $v_i(A_i \setminus \{e\}) \geq v_j(A_j \setminus \{e\})$ for some item $e \in A_i \cup A_j$.
\end{definition}

\begin{definition}[Envy-freeness]

An allocation ${\bf A}$ is envy-free (EF) if for any two agents $i,j \in N$, $v_i(A_i) \geq  v_i(A_j)$.   
\end{definition}

Specifically, Let $\boldsymbol{p} =(p_1, \ldots, p_n)$ be the payment vector, where each agent $i$ receives a payment $p_i \geq 0$.

\begin{definition}[Equitablility with Payment]

An allocation with payment $({\bf A}, \boldsymbol{p})$ is said to be equitable if for any two agents $i,j \in N$, $v_i(A_i) + p_i = v_j(A_j) + p_j$.
    
\end{definition}

\begin{definition}[Envy-freeness with Payment]

An allocation with payment $({\bf A}, \boldsymbol{p})$ is said to be envy-free if for any two agents $i,j \in N$,  $v_i(A_i) + p_i \geq  v_i(A_j) + p_j$.
    
\end{definition}

\begin{definition}[Equitable-convertible]
\label{def: ec}
An allocation ${\bf A}$ is said to be equitable-convertible if there exists a payment vector $\boldsymbol{p}$ such that $({\bf A}, \boldsymbol{p})$ is EQ.   
\end{definition}
\begin{definition}[Envy-freeable]
\label{def: efb}
An allocation ${\bf A}$ is called envy-freeable if there exists a payment vector $\boldsymbol{p}$ such that $({\bf A}, \boldsymbol{p})$ is EF.   
\end{definition}

\begin{definition}[Equitable-convertible and Envy-freeable]
\label{def: ec_efb}
An allocation ${\bf A}$ is both equitable-convertible and envy-freeable if there exists a payment vector $\boldsymbol{p}$ such that $({\bf A}, \boldsymbol{p})$ is both EQ and EF.   
\end{definition}

\begin{definition}[Utilitarian Social Welfare]
Given an allocation ${\bf A}$, its utilitarian social welfare ($USW$) is defined as the sum of values that the agent derives from the allocation, i.e., $USW({\bf A}) = \sum_{i \in N}v_i(A_i)$. 
For any $\alpha \in [0,1]$, we say that an allocation $\textbf{A}$ is $\alpha$-max-USW if it holds that $USW(\textbf{A}) \geq \alpha \cdot \max_{\textbf{B} \in \Pi_{n}(M)} USW (\textbf{B})$.
\end{definition}
\begin{definition}[Nash Welfare]
Given an allocation ${\bf A}$ in the goods setting, its Nash welfare ($NSW$) is defined as the geometric mean of values that the agent derives from the allocation, i.e., $NSW({\bf A}) = (\prod_{i \in N}v_i(A_i))^{\frac{1}{n}}$. 
For any $\alpha \in [0,1]$, we say that an allocation $\textbf{A}$ is $\alpha$-max-NSW if it holds that $NSW(\textbf{A}) \geq \alpha \cdot \max_{\textbf{B} \in \Pi_{n}(M)} NSW (\textbf{B})$.
\end{definition}

\section{Bounding Subsidy for Equitable-Convertible Allocations}
\label{sec: subsidy_equitability}
In this section, we study the required minimum subsidy to achieve equitability in the goods, chores, and mixed items settings.
We examine two distinct scenarios: one in which we are provided with an (equitable-convertible) allocation, and the other in which we can select such an allocation to minimize the required subsidy.

\subsection{When the Allocation Is Given}
 
First, we note that every allocation is equitable-convertible.
i.e., for an arbitrary allocation, we can always find a feasible payment vector that enables this allocation to be EQ.
To illustrate this, we present the following lemma to characterize the optimal payment vector when an allocation is given.
\begin{lemma}
\label{lem: equitable_minimum_payment}
  For any allocation ${\bf A}$, 
  let $\boldsymbol{p}^*$ be a payment vector such that $p_i^* = \max_{j \in N}v_j(A_j) - v_i(A_i)$. Then (1) $({\bf A}, \boldsymbol{p}^*)$ is EQ, and (2) this payment vector minimizes the total subsidy required for achieving EQ.

\end{lemma}
\begin{proof}
For the first claim, let $j^{max} \in \arg \max_{j \in N}v_j(A_j)$.
 For any agent $i \in N$, we have $p_i^* = v_{j^{max}}(A_{j^{max}}) - v_i(A_i)$. Rearranging this, we get $v_i(A_i) + p_i^* = v_{j^{max}}(A_{j^{max}})$, which means that this allocation $({\bf A}, \boldsymbol{p}^*)$ is EQ.
 
For the second claim, it suffices to show that for any feasible payment vector $\boldsymbol{p}$, we have $p_i^* \leq p_i$ for every agent $i \in N$. 
Suppose, for the contradiction, that there exists a feasible payment vector $\boldsymbol{p}^{\prime}$ such that $p_{i^{\prime}}^{\prime} < p_{i^{\prime}}^{*} = v_{j^{max}}(A_{j^{max}}) - v_{i^{\prime}}(A_{i^{\prime}}) $ for some agent $i^{\prime} \in N$. 
Then, we have $v_{i^{\prime}}(A_{i^{\prime}}) + p_{i^{\prime}}^{\prime} < v_{i^{\prime}}(A_{i^{\prime}}) + v_{j^{max}}(A_{j^{max}}) - v_{i^{\prime}}(A_{i^{\prime}}) = v_{j^{max}}(A_{j^{max}}) \leq v_{j^{max}}(A_{j^{max}}) + p_{j^{max}}^{\prime}$, which contradicts that $\boldsymbol{p}^{\prime}$ is a feasible payment vector.
\end{proof}

Next, we present the following proposition to demonstrate the required minimum subsidy for a specific allocation in the worst case when items are goods or chores.

\begin{proposition}
\label{prop: allocation_given_eq_minmum_subsidy}
     Given an allocation ${\bf A}$, in the worst case, the minimum required subsidy is $(n-1)m$ for the goods setting or $m$ for 
 the chores setting to achieve EQ. 
\end{proposition}

\begin{proof}
First, we consider the goods setting.
    For the lower bound, consider an instance with $n$ agents who have identical additive valuations and $m$ items, where for any item $e$, all agents value it at 1.
  Let ${\bf A}$ be an allocation where agent $i^*$ obtains all items.
  By Lemma \ref{lem: equitable_minimum_payment}, we can derive the optimal payment vector $\boldsymbol{p}$, where $p_{i^*} = 0$ and $p_{i} = m$ for any $i \neq i^*$.
  Therefore, we need at least $(n-1)m$ subsidy.
    Regarding the upper bound, let ${\bf A}$ be any allocation. We have $\max_{i\in N}v_i(A_i) - \min_{j \in N}v_j(A_j) \leq m$ since there are at most $m$ items whose total value is at most $m$ in one bundle. 
    By Lemma \ref{lem: equitable_minimum_payment} and the fact that the agent whose bundle has the highest value must receive zero payment, the required minimum subsidy is at most $(n-1)m$.

    Next, we consider the chores setting.
    For the lower bound, consider an instance with $n$ agents who have identical additive valuations and $m$ items, where for any item $e$, all agents value it as $-1$.
    Let ${\bf A}$ be an allocation where agent $i^*$ receives all items.
    By Lemma \ref{lem: equitable_minimum_payment}, we can get the optimal payment vector $\boldsymbol{p}$, where $p_{i^*} = m$ and $p_i = 0$ for any $i \neq i^*$.
    Therefore, we need at least $m$ subsidy.
    As for the upper bound, 
    for any allocation $\bf{A}$,  we have $\max_{i\in N}v_i(A_i) \leq 0$ since all items are chores.
    In this case,  we need to pay money for agents except the agent whose bundle with the highest value, and the required minimum subsidy is at most $m$, which can guarantee that the sum of the value of each agent's bundle and the money that she receives equals to $\max_{i \in N}v_i(A_i)$. 
\end{proof}

\subsection{When the Allocation Can Be Chosen}

In this part, we are allowed to choose an allocation.
Note that computing the required minimum subsidy is NP-hard. 
That is because checking whether a zero subsidy is required is equivalent to checking whether EQ allocations exist, which is NP-hard for identical additive valuations when all items are goods \cite{bouveret2008efficiency}.

Recall that when an allocation is given, the required subsidy in the worst case is $(n-1)m$ for the goods setting and $m$ for the chores setting. The following proposition demonstrates the lower bounds of the required subsidy if we can choose a desired allocation.

\begin{proposition}
\label{prop: lower_bound_subsidy_allocation_can_be_chosen}
For the goods or chores setting, there exists an instance where the minimum required subsidy is exactly $(n-1)$ to achieve EQ.
    For the mixed items setting, there exists an instance where the minimum required subsidy is exactly $\max \{n-1, m\}$ to achieve EQ.     
\end{proposition}

\begin{proof}
    First, we consider the goods setting.
     Consider an instance with $n$ agents who have additive valuations and one item that holds a value of 1 for each agent.
     Now, let ${\bf A}$ be an arbitrary allocation ${\bf A}$, where agent $i^*$ receives this item. 
    By Lemma \ref{lem: equitable_minimum_payment}, the optimal payments are $p_{i^*} = 0$ and $p_i = 1$ for any $i \neq i^*$. 
    Therefore, the required minimum subsidy is at least $(n-1)$.
    
    Next, we consider the chores setting.
    Consider an instance with $n$ agents who have identical additive valuations and $(n-1)$ items, where each item has the value of -1 for any agent.  
    Consider an arbitrary allocation ${\bf A}$, let $N^*$ denote the set of agents whose bundles are not empty. 
    Note that we have $N \setminus N^* \neq \emptyset$ since the number of items is smaller than the number of agents.
    Thus, we only need to give money to agents from $N^*$, and the required minimum subsidy is at least $(n-1)$.
    
    At last, we consider the mixed items setting.
    If $m < n-1$, by the above instances where items are goods, we can derive that the lower bound is at least $(n-1)$.
    If $m \geq n-1$,
    we consider an instance where each agent has the additive valuation.
    For item $e_1$, all agents value it at 1, and for the other items $e \neq e_1$, all agents value it at -1.
    To minimize the required subsidy, agent 1's bundle should be empty.
    That is because she values all items positively, and the value of her bundle is the highest among all agents. Specifically, for arbitrary allocation $\textbf{A}$, we have $\max_{j\in N}v_i(A_i) = v_1(A_1)$.
    If she receives some items, by Lemma \ref{lem: equitable_minimum_payment}, she will receive a payment of $0$, but we need to allocate money to the other agents.
    Thus, in the allocation that minimizes the subsidy required, we should allocate items to the other agents except for agent 1, and the total subsidy required is at least $m$. Considering the above two cases, the claim holds. 
\end{proof}

Then, a natural question is whether we can find an allocation whose required minimum subsidy is at most $(n-1)$ when items are goods or chores, or $\max\{n-1,m\}$ when there is the mixed items setting.
We answer this question affirmatively for the goods or chores setting, and show that for the mixed items setting with additive valuations, the required minimum subsidy is at most $(m+n-2)$.
Before presenting our algorithm, we examine what is the upper bound of the required minimum subsidy for the special allocation that has already limited inequitability, such as an EQ1 allocation.

\begin{proposition}
\label{prop: eq1_subsidy}
     For any EQ1 allocation ${\bf A}$, the required minimum subsidy is at most $(n-1)$ to achieve EQ.
\end{proposition}

\begin{proof}
    Consider an arbitrary EQ1 allocation ${\bf A}$. 
    By the definition of EQ1 (Definition \ref{def: eq1}), for any two agents $i,j \in N$ with $A_i, A_j \neq \emptyset$, we have $v_i(A_i \setminus \{e\}) \geq v_j(A_j \setminus \{e\})$ for some item $e \in A_i \cup A_j$.
    Fix two arbitrary agents $i, j \in N$,
    Since $\lvert \Delta_{\ell}(A_{\ell} \setminus \{e\},e)\rvert \leq 1$ for any $\ell \in N$ and $e \in A_\ell$, we get $v_j(A_j) - v_i(A_i) \leq 1$, which means that the value gap between any two bundles is at most one.
    By Lemma \ref{lem: equitable_minimum_payment} and the above property, the minimum payment for agent $i \in N$ is $\max_{j \in N}v_j(A_j)-v_i(A_i) \leq 1$.
    Note that the agent whose bundle has the largest value must receive zero payment.
    Therefore, the required minimum subsidy is at most $(n-1)$.
\end{proof}

\begin{algorithm}[tb]
\caption{Squeeze Algorithm}
\label{alg: nonmonotone_subsidy}
\KwIn{An instance $\mathcal{I} = \langle M, N, \boldsymbol{v} \rangle$ with the goods, chores, mixed item setting with objective or additive valuations.}
\KwOut{An allocation $\mathcal{{\bf A}}$ with bounded subsidy for equitability}

Let ${\bf A} = (\emptyset, \ldots, \emptyset)$ and $M^{\prime} = M$;

Let $i^{max} \in \arg \max_{i \in N}v_i(A_i)$ and $i^{min} \in \arg\min_{i \in N}v_i(A_i)$ (breaking ties arbitrarily);

Let $M^-_i = \{e | e\in M^{\prime}, \Delta_{i}(A_{i},e) \leq 0\}$ and 
$M^+_i = \{e | e\in M^{\prime}, \Delta_{i}(A_{i},e) \geq 0\}$;

\While{$M^{\prime} \neq \emptyset$}{

\uIf {$ \exists e \in M^{\prime}$ such that $\Delta_{i^{max}}(A_{i^{max}},e) \leq 0 $}
{

Let $e^{*} \in \arg\max_{e \in M^-_{i^{max}}}\Delta_{i^{max}}(A_{i^{max}},e)$ (breaking ties arbitrarily);

$A_{i^{max}}$ = $A_{i^{max}} \cup \{e^{*}\}$ and $M^{\prime} = M^{\prime} \setminus \{e^{*}\}$; \label{lin: agent_maximum_value_picks_item}
}
{
\uElseIf{$ \exists e \in M^{\prime}$ such that $\Delta_{i^{min}}(A_{i^{min}},e) \geq 0 $ }{

Let $e^{*} \in \arg\max_{ e \in M^+_{i^{min}}}\Delta_{i^{min}}(A_{i^{min}},e)$ (breaking ties arbitrarily);

$A_{i^{min}}$ = $A_{i^{min}} \cup \{e^{*}\}$ and $M^{\prime} = M^{\prime} \setminus \{e^{*}\}$; \label{lin: agent_minium_value_picks_item}

}
\Else{
break; \label{lin: break}
}
}

}

\uIf{$M^{\prime} \neq \emptyset$}{
$A_{i^{min}} = A_{i^{min}} \cup  M^{\prime} $;
}

\Return ${\bf A}$;

\end{algorithm}
 We design the following algorithm, named the Squeeze Algorithm, where the details of our algorithm can be found in Algorithm \ref{alg: nonmonotone_subsidy}.
The high-level idea of our algorithm is that in each iteration, we first attempt to allocate the item with the highest negative marginal value to the agent whose bundle currently has the highest value. This strategy is motivated by Lemma \ref{lem: equitable_minimum_payment}, which establishes that for agent $i$, the required minimum payment is $p_i = \max_{j \in N}v_j(A_j) - v_i(A_i)$. Intuitively, reducing $\max_{j \in N}v_j(A_j)$ helps us minimize the subsidy payments for all agents except for the one whose bundle has the highest value. If no item exhibits a negative marginal value for the agent whose bundle has the highest value, we then consider allocating the item that has the highest positive marginal value to the agent whose bundle has the minimum value, since increasing $v_i(A_i)$ can reduce the subsidy payment required for that agent.
Finally, if neither of the above strategies is feasible, we directly allocate the remaining unallocated items to the agent whose bundle has the lowest value to ensure the complete allocation of all items.
To handle the exponential input size efficiently, the valuations are presented to the algorithm as “valuation oracles” – black boxes that can be queried for the valuation of a set $S$, returning its valuation $v(S)$.

\begin{definition}[Objective Valuations \cite{barman2024nearly}]
 The valuation functions $\boldsymbol{v}$ are objective, if for each item $e \in M$, we have $ \Delta_{i}(S,e) \geq 0$  for any agent $i \in N$,  and any subset $S \subseteq M \setminus \{e\}$, or $\Delta_{i}(S,e) \leq 0 $  for any agent $i \in N$,  and any subset $S \subseteq M \setminus \{e\}$.

\end{definition}
\begin{theorem}
\label{the: nonomontone_subsidy}
For the goods, chores, or mixed items setting with objective valuations, Algorithm \ref{alg: nonmonotone_subsidy} computes an allocation that satisfies EQ1 and requires at most $(n-1)$ minimum subsidy to achieve EQ in polynomial time.
For the mixed items setting with additive valuations,
    Algorithm \ref{alg: nonmonotone_subsidy} computes an allocation that requires at most $(m+n-2)$ minimum subsidy to achieve EQ in polynomial time.
\end{theorem}

\begin{proof}
We first show that when the while loop terminates,  we have $\max_{j \in N} v_j(A_j) - v_i(A_i) \leq 1 $ for any agent $i \in N$, and this partial allocation is EQ1.
We prove it by induction.
For the base case, the property trivially holds for the empty allocation.
From the induction hypothesis, the partial allocation $\textbf{A}^{k-1} = (A_1^{k-1}, \ldots, A_n^{k-1})$ obtained at the end of the $(j-1)$th iteration satisfies the desired property.
We will show that after the $j$th iteration, the property still holds for allocation $\textbf{A}^{k} = (A_1^{k}, \ldots, A_n^{k})$.
In the $j$th iteration, if the while loop is terminated in line \ref{lin: break}, the partial allocation remains the same, and the property trivially holds. 
If not, let us consider the following two cases.

\textbf{Case 1:} The agent whose bundle has the maximum value is chosen to pick a new item $e^*$ which has a non-positive marginal value in the line \ref{lin: agent_maximum_value_picks_item}. 
Note that, except for agent $i^{max}$'s bundle, other agents' bundles remain unchanged, i.e., $A_{i}^{k-1} = A_i^{k}$ for any $i \neq i^{max}$.
First, we show that  $\max_{j \in N} v_j(A_j^{k}) - v_i(A_i^{k}) \leq 1 $.
If the value of agent $i^{max}$'s new bundle is still the maximum one, it is trivial that the property still holds.
If not, suppose that agent $\ell$'s bundle has the largest value.
Then, we have $\max_{j \in N}v_j(A_j^{k}) - v_{i^{max}}(A_{i^{max}}^k) = v_{\ell}(A_{\ell}^{k})-v_{i^{max}}(A_{i^{max}}^{k-1} \cup \{e^*\}) \leq v_{\ell}(A_{\ell}^{k-1})-[v_{i^{max}}(A_{i^{max}}^{k-1})-1] \leq 1$, where the first inequality follows from that the marginal value of any item is at least -1.
Next, we show that the partial allocation $\textbf{A}^k$ is EQ1.
For agent $i^{max}$, we have $v_{i^{max}}(A_{i^{max}}^k \setminus \{e^*\}) \geq v_{i^{max}}(A_{i^{max}}^{k-1}) \geq v_i(A_i)$ for any $i \neq i^{max}$.
For any agent $i \neq i^{max}$, we have $v_i(A_{i}^{k-1} \setminus \{e\}) \geq v_{i^{max}}(A_{i^{max}}^{k-1} \setminus \{e\})$ for some $e \in A_i \cup A_{i^{max}}$ and $v_{i^{max}}(A_{i^{max}}^{k-1}) \geq v_{i^{max}}(A_{i^{max}}^{k})$.
Therefore, EQ1 still holds for $\textbf{A}^{k}$.

\textbf{Case 2:} The agent whose bundle has the minimum value is chosen to pick a new item $e^*$ which has nonnegative marginal value in the line \ref{lin: agent_minium_value_picks_item}.
Note that, except for agent $i^{min}$'s bundle, other agents' bundles remain unchanged, i.e., $A_{i}^{k-1} = A_i^{k}$ for any $i \neq i^{min}$.
First, we show that  $\max_{j \in N} v_j(A_j^{k}) - v_i(A_i^{k}) \leq 1 $.
If the value of agent $i^{min}$'s new bundle is the maximum one, i.e., $\max_{j \in N} v_j(A_j^{k}) = v_{i^{min}}(A_{i^{min}}^k)$, 
for any agent $i \neq i_{min}$, we have $\max_{j \in N} v_j(A_j^{k})-v_i(A_i^k) = v_{i^{min}}(A_{i^{min}}^k) - v_i(A_i^k) = v_{i^{min}}(A_{i^{min}}^{k-1} \cup \{e^*\}) - v_i(A_i^{k-1}) \leq v_{i^{min}}(A_{i^{min}}^{k-1})+1 - v_i(A_i^{k-1}) \leq 1$, where the third inequality holds for the fact that the marginal value of each item does not exceed 1.
If not, for agent $i^{min}$, we have $\max_{j \in N}v_j(A_j^k) - v_{i^{min}}(A_{i^{min}}^k) \leq \max_{j \in N}v_j(A_j^{k-1}) - v_{i^{min}}(A_{i^{min}}^{k-1}) \leq 1$ where the first inequality follows from that $\max_{j \in N}v_j(A_j^k) \geq v_{i^{min}}(A_{i^{min}}^k)\geq  v_{i^{min}}(A_{i^{min}}^{k-1})$.
Next, we show that EQ1 holds for $\textbf{A}^k$.
For agent $i^{min}$, we have $v_{i^{min}}(A_{i^{min}}^{k}) \geq v_{i^{min}}(A_{i^{min}}^{k-1})$ and $ v_{i^{min}}(A_{i^{min}}^{k-1} \setminus \{e\}) \geq v_i(A_{i}^{k-1} \setminus \{e\})$ for some $e \in A_{i^{min}} \cup A_i$.
For agent $i \neq i^{min}$, we have $v_i(A_i^{k-1}) \geq v_{i^{min}}(A_{i^{min}}^{k-1}) = v_{i^{min}}(A_{i^{min}}^{k} \setminus \{e^*\})$.
Thus, the partial allocation $\textbf{A}^k$ is still EQ1.

Overall, by induction, the properties still hold after the $j$th iteration, and the proof is established. 
For goods, chores, or mixed items setting with objective valuations, $M^{\prime}$ must be empty after the while loop terminates. 
That is because, by the definition of monotone valuations, or objective valuations, all agents value one item as a good or a chore. 
Thus, the total required minimum subsidy is at most $(n-1)$.
For mixed items setting with additive valuations, there are two cases.
If $M^{\prime}$ is empty, we can also conclude that the total subsidy required is at most $(n-1)$. 
If not, the remaining items are allocated to the agent whose bundle has the minimum value, and we need at most extra $\lvert M^{\prime} \rvert$ money to compensate for agent $i^{min}$.
Then, the total subsidy required is at most $n-1 + \lvert M^{\prime} \rvert \leq n-1+m-1=m+n-2$.
Therefore, combining the above two cases, it can be concluded that the total required minimum subsidy is at most $(m+n-2)$.
\end{proof}

\section{Bounding Subsidy for Both Equitable-Convertible and Envy-Freeable Allocations}
\label{sec: subsidy_eq_ef}
In this section, we study whether an allocation can satisfy EQ and EF with the same payment vector.
There are well-established characterizations for allocations that are envy-freeable \cite{halpern2019fair}.
Because we consider two criteria, it is clear that there exist some allocations that cannot be EQ and EF simultaneously, even if we provide  some money to agents. 
First, we will characterize an allocation that is both equitable-convertible and envy-freeable with the same payment vector and present the following sufficient and necessary conditions.
\begin{lemma}
\label{the: characterization_ef_eq}
    For an allocation ${\bf A}$, the following statements are equivalent: (1) ${\bf A}$ is both equitable-convertible and envy-freeable with the same payment vector, and (2) $v_j(A_j) \geq v_i(A_j)$ for any two agents $i,j \in N$.
\end{lemma}

\begin{proof}
First, we show that $(1) \Rightarrow (2)$. Suppose an allocation ${\bf A}$ is equitable-convertible and envy-freeable. In that case, it means that there exists a payment vector $\boldsymbol{p}$, such that for any two agents $i, j \in N$, we have $v_i(A_i) + p_i = v_j(A_j) + p_j$ and $v_i(A_i) +p_i \geq v_i(A_j) + p_j$.
Combining the above two inequalities together, we get $v_j(A_j) \geq v_i(A_j)$.

Next, we show that $(2) \Rightarrow (1)$. When the claim (2) is satisfied, it suffices to show that there exists a payment vector $\boldsymbol{p}$ to enable ${\bf A}$ to be both equitable-convertible and envy-freeable.
Let $p_i = \max_{k \in N}v_k(A_k) - v_i(A_i)$ for any agent $i \in N$.
Note that for any agent $i \in N$, we have $v_i(A_i) + p_i = \max_{k \in N}v_k(A_k)$, which means that ${\bf A}$ is equitable.
Then, fix any two agents $i, j \in N$, we have $v_i(A_i) + p_i = \max_{k \in N}v_k(A_k) \geq \max_{k\in N} v_k(A_k)+v_i(A_j) - v_j(A_j) = v_i(A_j) + p_j$, where the inequality follows from claim $(2)$. Thus, ${\bf A}$ is also envy-freeable.
\end{proof}

Next, we examine the required minimum subsidy if an equitable-convertible and envy-freeable allocation is given.
Recall from the proof of Proposition \ref{prop: allocation_given_eq_minmum_subsidy}, where we focus on the required minimum subsidy only for EQ in the worst case when items are goods or chores, 
two constructed allocations are actually both equitable-convertible and envy-freeable. 
That is because all agents have the same valuation $v$, and then $v_j(A_j) = v(A_j) = v_i(A_j)$ holds for any two agents $i, j \in N$, which is the requirement of an equitable-convertible and envy-freeable allocation in Theorem \ref{the: characterization_ef_eq}.  
Therefore, we derive the following corollary directly.
\begin{corollary}
\label{coro: susbsidy_ef_eq_worst_case}
  Given an equitable-convertible and envy-freeable allocation ${\bf A}$, in the worst case, to achieve both EQ and EF, the required minimum subsidy is $(n-1)m$ for the goods setting, or $m$ for the chores setting.   
\end{corollary}
Lastly, we consider the scenario where the allocation can be chosen, and we present the following proposition to show the lower bound of the minimum subsidy to achieve both EQ and EF.
 
\begin{proposition}
\label{prop: subsidy_ef_eq_lower_bound}
 For the goods setting, there exists an instance where the required minimum subsidy is $(n-1)m$ to achieve both EQ and EF.
 For the chores setting, there exists an instance where the required minimum subsidy is  $m$ to achieve both EQ and EF. 
\end{proposition}
\begin{proof}
For the goods setting, consider an instance with $n$ agents who have additive valuations and $m$ items. 
For agent $1$, the value of each item is $1$.
For the remaining agents $i \neq 1$, the value of each item is $(1-\epsilon)$, where  $0<\epsilon \leq 1$.
By Theorem \ref{the: characterization_ef_eq}, in this instance, the only equitable-convertible and envy-freeable allocation ${\bf A}$ is that agent $1$ picks all items since for any item $e \in M$, we have $v_1(e) > v_i(e)$ for any $i \neq 1$. 
Then, by Lemma \ref{lem: equitable_minimum_payment}, it is not hard to see that the required minimum subsidy is $(n-1)m$.

For the chores setting, consider an instance with $n$ agents who have additive valuations and $m$ items. 
For agent $1$, the value of each item is $(-1+\epsilon)$, where $\epsilon\rightarrow 0$.
For the remaining agent $i \neq 1$, the value of each item is $-1$.
By Theorem \ref{the: characterization_ef_eq}, in this instance, the only equitable-convertible and envy-freeable allocation ${\bf A}$ is that agent $1$ picks all items since for any item $e \in M$, we have $v_1(e) > v_i(e)$ for any $i \neq 1$. 
Then, by Lemma \ref{lem: equitable_minimum_payment}, it is not hard to see that the required minimum subsidy is $m(1-\epsilon)\rightarrow m$.
\end{proof}

In contrast to providing the minimum subsidy only for equitability, this lower bound of minimum subsidy for equitability and envy-freeness is the same as when the allocation is specified. 
This indicates that if a specific allocation fulfills the condition outlined in Theorem \ref{the: characterization_ef_eq}, then the minimum subsidy is ``optimally bounded''.

\section{Bounded Subsidy for Fairness with Efficiency}

In this section, we investigate whether an allocation can achieve fairness with an "optimally bounded" minimum subsidy while simultaneously guaranteeing some level of efficiency, such as utilitarian social welfare or Nash welfare.
We focus on the goods setting, where social welfare calculations exclude the subsidy payments made to agents.

\subsection{Equitability with Efficiency Guarantees}
In this part, we consider only one fairness notion - equitability.
By Proposition \ref{prop: lower_bound_subsidy_allocation_can_be_chosen}, the lower bound of the required minimum subsidy for equitability is $(n-1)$ when the allocation can be chosen.
Therefore, we only consider the allocation that needs a minimum subsidy of at most $(n-1)$.
However, we have a negative result to show that in some instances, we cannot find any allocation that needs at most $(n-1)$ minimum subsidy and guarantees some non-zero approximation of maximum utilitarian social welfare or maximum Nash welfare. 

\begin{proposition}
\label{the: subsidy_no_efficiency_scaled}
    For any $\alpha>0$, there exists a goods instance where any allocation whose required minimum subsidy is at most $(n-1)$ to achieve EQ cannot guarantee $\alpha$-max-USW or $\alpha$-max-NSW.
\end{proposition}

\begin{proof}
\begin{table}[tb]
    \centering
    \setlength{\tabcolsep}{1.5mm}{
    \begin{tabular}{c|c|c|c|c|c|c|c}
    \toprule
      & $e_1$ & $e_2$ & $\ldots$& $e_{n}$ & $e_{n+1}$ & $\ldots$  & $e_{m}$ \\
      \midrule
        agent 1 & 1 & 1 &  $\ldots$ & 1 & 1 & \ldots &1 \\
        \midrule
        agent 2 & $0$ & $\epsilon$ & $\ldots$ &$0$ & $0$ & \ldots & $0$\\
        \midrule
        $\ldots$ & $\ldots$ & $\ldots$ & $\ldots$ & $\ldots$ &$\ldots$  & $\ldots$ & $\ldots$ \\
        \midrule
        agent $n$ & $0$ & $0$ & $\ldots$ &$\epsilon$ & $0$ & \ldots& $0$\\
        \bottomrule
    \end{tabular}  
    }
    \caption{Example showing that any allocation with ``optimally bounded'' subsidy has an arbitrarily poor approximate maximum utilitarian social welfare or maximum Nash welfare.}
    \label{tab: subsidy_no_efficiency}
    \end{table}
    Let us consider the following instance with $n$ agents who have additive valuation functions and $m$ items, where $m \gg n$.
    The value of each item is shown in Table \ref{tab: subsidy_no_efficiency}.
    For agent 1, all items have the value of 1.
    For agent $i \geq 2$, only item $i$ has the value of $\epsilon$, and the remaining items have the value of 0, where $0<\epsilon<\frac{1}{n}$.

    First, let us consider the utilitarian social welfare.
    In this instance, the maximum utilitarian social welfare is $m$, which is achieved when all items are allocated to agent 1.
    In the maximum utilitarian social welfare allocation, the required minimum subsidy is $(n-1)m > n-1$.
    Then, consider any allocation ${\bf A}$ that needs at most $(n-1)$ subsidy.
    In allocation ${\bf A}$, either agent 1 picks one item or her bundle is empty.
    Otherwise, the required subsidy will exceed $(n-1)$.
    Thus, we have $SW({\bf A}) \leq 1+(n-1)\epsilon <2$.
    For any $\alpha>0$, let $m>\frac{2}{\alpha}$, we have $\frac{2}{m}<\alpha$.
    Thus, no feasible allocation whose required minimum subsidy is at most $(n-1)$ can guarantee $\alpha$-max-USW.

    Next, let us consider Nash welfare.
    In this instance, the maximum Nash welfare is $[(m-n+1)\epsilon^{n-1}]^\frac{1}{n}$, which is achieved when agent $i\geq 2$ picks the item $i$, and the remaining items are allocated to agent 1.
    In the maximum Nash welfare allocation, the required minimum subsidy is $(n-1)(m-n+1-\epsilon) > n-1$.
    Similarly, consider any allocation $A$ that needs at most $(n-1)$ subsidy.
    In allocation ${\bf A}$, either agent 1 picks one item or her bundle is empty. 
    Thus, we have $NSW({\bf A})\leq\epsilon^{\frac{n-1}{n}}$.
    For any $\alpha>0$, let $m>(\frac{1}{\alpha})^n+n-1$, we have $[\frac{\epsilon^{n-1}}{(m-n+1)\epsilon^{n-1}}]^{\frac{1}{n}} < \alpha$.
     Thus, no feasible allocation whose required minimum subsidy is at most $(n-1)$ can guarantee $\alpha$-max-NSW.
\end{proof}

The above theorem demonstrates that we cannot guarantee any efficiency in the scaled valuation setting, even for additive valuations.
Then, we turn our attention to the normalized valuation setting, where $v_i(M)=1$ for any $i \in N$.
This setting is commonly studied in the study of fairness and efficiency \cite{barman2020optimal,sun2023equitability,bu2025approximability}.
In the normalized valuation setting, any feasible allocation requires at most $(n-1)$ minimum subsidy to achieve EQ.
That follows from the fact that for any allocation $\textbf{A}$, we have $\max_{j \in N}v_j(A_j) - v_i(A_i) \leq 1$, since the total value of items for any agent is at most 1.
Furthermore, the lower bound of minimum subsidy $(n-1)$ still holds, as shown in the proof of Proposition \ref{prop: lower_bound_subsidy_allocation_can_be_chosen}, the constructed instance for the goods setting has the normalized valuations.


Next, we investigate the following natural question.
In the normalized valuation setting, is it possible that we can find an allocation that satisfies ex-post fairness, e.g., EQ1, and guarantees some efficiency like utilitarian social welfare or Nash welfare, which achieves multiple targets? 
Here, we mainly consider the normalized subadditive valuation setting.

\paragraph{Utilitarian Social Welfare} We first study the utilitarian social welfare and present our approximation algorithm, the Balanced-Packing algorithm (Algorithm \ref{alg: eq1_subadditive_sw}).
Our algorithm consists of three main steps: (1) partition items into two categories - the high marginal value items and the low marginal value items; (2) allocate the high marginal value items to some agents who cannot receive additional items before step (3) and greedily fill up some ``empty bags'' for the remaining agents; (3) If unallocated items remain, we use Algorithm \ref{alg: nonmonotone_subsidy} to extend the partial allocation to a complete one while preserving EQ1; otherwise, we get the desired allocation.

\begin{theorem}
\label{the: usw_subadditive}
    For normalized subadditive valuations in the goods setting,  Algorithm \ref{alg: eq1_subadditive_sw} computes an EQ1 allocation with $O(\frac{1}{n})$-max-USW in polynomial time.
\end{theorem}
\begin{proof}
First, we show that the allocation computed by Algorithm \ref{alg: eq1_subadditive_sw}  has the utilitarian social welfare at least $\frac{n}{3n-3}$.
Let $M^{\prime}$ denote the set of remaining items and $N^{-}$ denote the set of agents who have not picked any item.
For the first while loop in line \ref{line: first while loop}, the algorithm allocates the items with the high marginal value that is at least $\frac{1}{3n-3}$ to some agents.
If the marginal value of one item is at least $\frac{1}{3n-3}$ to more than one agent, the agent who has the largest valuation picks this item. 
When the first loop terminates, for any item $e \in M^{\prime}$, we have $v_i(e) < \frac{1}{3n-3}$ for any $i \in N^{-}$, and let $\bf{B}$ denote the partial allocation.
Then, the algorithm checks whether $USW(\bf{B}) \geq  \frac{n}{3n-3}$.
If it holds, we are done.
If not, it means that $N^{-} \neq \emptyset$.
Since only the agent who has the highest valuation is chosen in line \ref{line: highest_value}, we have $v_i(B_j) \leq v_j(B_j)$ for any $i \in N^{-}$ and $j \in N^{+}$.
Accordingly, we get $\sum_{j \in N^{+}}v_i(B_j) \leq USW (\bf{B}) < \frac{n}{3n-3}$.
Next, the algorithm continuously adds the item to the empty set $S$ until there exists some agent $i \in N^{-}$ such that $v_i(S) > \frac{1}{3n-3}$ and then $S$ is allocated to agent $i$.
Repeat this process until $N^{-} = \emptyset$.
In order to guarantee that the above procedure is executed successfully, we need to show that for any agent $i \in N^{-}$, there are enough items in $M^{\prime}$ to pick, i.e., the value of remaining items in $M^{\prime}$ for agent $i \in N^{-}$ is still at least $\frac{1}{3n-3}$.
Consider an arbitrary iteration starting with $N^{-} \neq \emptyset$.
Let $S_j$ denote the bundle that agent $j \in N \setminus (N^{+} \cup N^{-})$ receives.
Note that for any agent  $i \in N^{-}$ and  $j \in N \setminus (N^{+} \cup N^{-})$, it must hold that $v_i(S_j) \leq \frac{2}{3n-3}$ since at the time when agent $j$ receives a bundle $S_j$, before the last item $e_{\ell}$ is added to $S_j$, we have $v_i(S_j \setminus \{e_{\ell}\}) < \frac{1}{3n-3}$. Then, by subadditivity, we get $v_i(S_j) \leq v_i(S_j \setminus \{e_{\ell}\}) + v_i(e_{\ell}) < \frac{2}{3n-3}$.
Consequently, for agent $i$, there are two cases.

\textbf{Case 1:} If $N^{+} \neq \emptyset$, we have
\begin{equation*}
\begin{aligned}
v_i(M^{\prime})
& = v_i(M \setminus \bigcup_{j \in N^{+}}B_j \setminus \bigcup_{j \in N \setminus(N^{+}\cup N^{-})}S_j)\\
& \geq v_i(M) - v_i(\bigcup_{j \in N^{+}}B_j) - v_i(\bigcup_{j \in N \setminus(N^{+}\cup N^{-})}S_j)\\
& \geq v_i(M) - \sum_{j \in N^{+}}v_i(B_j) -\sum_{j \in N \setminus(N^{+}\cup N^{-})}v_i(S_j) \\
& > 1 - \frac{n}{3n-3} -\frac{2}{3n-3}(n - \lvert N^{+} \rvert -\lvert N^{-} \rvert)\\
& \geq \frac{1}{3n-3}.
\end{aligned}
\end{equation*}   
where the second and third inequalities follow from the subadditivity, and the last one holds since $(\lvert N^{+}\rvert  + \lvert N^{-}\rvert) \geq 2$.

\textbf{Case 2:} If $N^{+} = \emptyset$, we have
\begin{equation*}
\begin{aligned}
v_i(M^{\prime})
& = v_i(M \setminus \bigcup_{j \in N \setminus  N^{-}}S_j)\\
& \geq v_i(M) - v_i(\bigcup_{j \in N \setminus N^{-}}S_j)\\
& \geq v_i(M) -\sum_{j \in N \setminus  N^{-}}v_i(S_j) \\
& > 1  -\frac{2}{3n-3}(n -\lvert N^{-} \rvert)\\
& > \frac{1}{3} \geq \frac{1}{3n-3}.
\end{aligned}
\end{equation*} 
where the second and third inequalities follow from the subadditivity, and the last one holds since $  \lvert N^{-}\rvert \geq 1$.
Thus, every agent $i \in N^{-}$ receives a bundle whose value is at least $\frac{1}{3n-3}$. 
When $N^{-} = \emptyset$, the utilitarian social welfare of the partial allocation is at least $\frac{n}{3n-3}$.
In the following steps, the utilitarian social welfare cannot decrease, so the utilitarian social welfare of output allocation must be at least $\frac{n}{3n-3}$. 
Given an instance with normalized subadditive valuations, $\sum_{i \in N}v_i(M) = n$ is a trivial upper bound on the maximum utilitarian social welfare.
Therefore, the returned allocation has an $O(\frac{1}{n})$-fraction of maximum utilitarian social welfare.

Next, we prove that the output allocation is EQ1.
Let $\bf{B}$ still be the partial allocation when the first while loop terminates.
Note that $\lvert B_i \rvert = 1$ for any agent $i \in N^{+}$ and $B_i = \emptyset$ for any agent $i \in N^{-}$.
Thus, that partial allocation $\bf{B}$ is trivially EQ1.
If $USW(\bf{B}) \geq \frac{n}{3n-3}$, the partial allocation remains unchanged before line $\ref{line: second_while_loop}$.
If not, the algorithm proceeds to allocate items to agents in $N^{-}$.
Let $\bf{C}$ be the partial allocation when $N^{-} = \emptyset$ and we claim that the partial allocation $\bf{C}$ is also EQ1.
For any agent $i \in N$, when comparing to agent $j \in N^{+}$, it does not violate EQ1 since $\lvert C_i \rvert = 1$, and when comparing to agent $j \in N \setminus N^{+}$, we have $v_i(C_i) \geq \frac{1}{3n-3} > v_i(C_j \setminus \{e_{\ell}\})$ where $e_{\ell}$ is the last item added to $C_j$, which means that it still EQ1. 
Thus, the partial allocation $\bf{C}$ satisfies EQ1.
If $M^{\prime} = \emptyset$, the returned allocation is EQ1. 
If not, the algorithm follows Algorithm \ref{alg: nonmonotone_subsidy} to allocate the remaining items, i.e., every time the item is allocated to the agent whose bundle has the lowest value.
By the proof of Theorem \ref{the: nonomontone_subsidy}, it is not hard to check that, in this process, EQ1 is not violated.
Therefore, we can conclude that the returned allocation is EQ1.
\end{proof}

\begin{algorithm}[tb]
\caption{Balanced-Packing Algorithm}
\label{alg: eq1_subadditive_sw}
\KwIn{A goods instance $\mathcal{I} = \langle M, N, \boldsymbol{v} \rangle$ with normalized subadditive valuations}
\KwOut{An allocation $\mathcal{{\bf A}}$ with ``optimally bounded'' subsidy and $O(\frac{1}{n})$-max-USW }

Let ${\bf A} = (\emptyset, \ldots, \emptyset)$,  $M^{\prime} = M$, $N^{-} = N$ and $N^{+} = \emptyset$;

\While{$\exists e^* \in M^{\prime}$ such that $v_{i}(e^*) \geq \frac{1}{3n-3}$ for some $i \in N^{-}$ 
\label{line: first while loop}}{
Let $i^{max} \in \arg\max_{i \in N^{-}}v_i(e^*)$; \label{line: highest_value}

$N^{+} = N^{+} \cup \{i^{max}\}$ and $N^{-} = N^{-} \setminus \{i^{max}\}$;

$A_{i^{max}} = \{e^*\}$ and $M^{\prime}  = M^{\prime} \setminus \{e^*\}$;
}

\uIf{$USW({\bf A}) < \frac{n}{3n-3}$}{


\While{$N^{-}\neq \emptyset$}{
Let $S = \emptyset$;

\While{$\forall i \in N^{-}$, $v_i(S) < \frac{1}{3n-3}$}{
Pick an arbitrary item $e \in M^{\prime}$;

$S = S \cup \{e\}$ and $M^{'} = M^{'} \setminus \{e\}$; 

}

Let $N^{*} = \{i ~|~ i \in N^{-}~ \textit{and}~ v_i(S) \geq \frac{1}{3n-3} \}$ and choose an arbitrary agent $i \in N^{*}$;

$A_i = S$ and $N^{-} = N^{-} \setminus \{i\}$;
}
}

\uIf{$M^{\prime} \neq \emptyset$ \label{line: second_while_loop}}{



Execute Algorithm \ref{alg: nonmonotone_subsidy} to extend the partial allocation ${\bf A}$ to a complete one;

}

\Return ${\bf A}$;

\end{algorithm}

Regarding the upper bound of the utilitarian social welfare that an EQ1 allocation has, \cite{sun2023equitability} shows that no EQ1 allocation achieves an $O(\frac{1}{n})$-max-USW when agents have normalized additive valuations.
Therefore, it can be concluded that our algorithm provides an asymptotically tight fraction of maximum utilitarian social welfare that EQ1 allocations guarantee when agents have normalized subadditive valuations.

\paragraph{Nash Welfare} We now turn attention to the Nash welfare for two agents.
Unlike envy-based notions for two agents, where classic techniques such as divide-and-choose are applicable, such approaches cannot be directly applied to EQ1 allocations.
Consequently, we require a subtle case analysis for item allocation.
Our algorithm's intuition builds upon the framework established in Algorithm \ref{alg: eq1_subadditive_sw}. 
We first identify the items with high marginal values. 
If such an item exists, we allocate it first.
If not, we greedily add items in an empty bag until one agent thinks its value is sufficient,  and allocate it to her.
Finally, we execute Algorithm \ref{alg: nonmonotone_subsidy} to extend the partial allocation to a complete one.
The complete algorithm details are provided in Algorithm \ref{alg: eq1_nsw_computation}.

\begin{theorem}
\label{the: subadditive_two_agents_nash_welfare}
    For two agents with normalized subadditive valuations in the goods setting,  Algorithm \ref{alg: eq1_nsw_computation} computes an EQ1 allocation that has $\frac{1}{3}$-max-NSW in polynomial time.
\end{theorem}
\begin{proof}
The algorithm is executed in two phases. 
In the first phase, we get a partial EQ1 allocation with some guarantees of Nash welfare.
In the second phase, the partial EQ1 allocation is extended to a complete EQ1 allocation by Algorithm \ref{alg: nonmonotone_subsidy}, and it is not hard to see that the Nash welfare does not decrease in this phase since we consider the goods setting.
Therefore, it suffices to show that after the first phase, the partial allocation is EQ1, and its Nash welfare is at least $\frac{1}{3}$-max-NSW.
Let us consider the following cases.

\textbf{Case 1.} If there exists two different items $e_1$ and $e_2$ such that $v_1(e_1) \geq\frac{2}{3}$ and $v_2(e_2) \geq \frac{2}{3}$, item $e_1$ is allocated agent 1 and item $e_2$ is allocated to agent 2.
It is trivial that this partial allocation is EQ1, and its Nash welfare is at least $\frac{2}{3}$-max-NSW.

\textbf{Case 2.} If there exists one item $e^*$ such that $v_i(e^*) \geq \frac{2}{3}$ for some $i \in [2]$.
Without loss of generality, we assume that agent 1 is such an agent.
Then, let us consider the following two subcases.
\begin{itemize}
    \item If $v_1(e^*) \geq v_2(M \setminus \{e^*\})$, item $e^*$ is allocated to agent 1 and $M \setminus \{e^*\}$ is allocated to agent 2.
    It is not hard to see that this partial allocation is EQ1.
    Since $v_2(e^*) \leq \frac{2}{3}$, we have $v_2(M \setminus \{e^*\}) \geq v_2(M) -v_2(e^*) \geq 1-\frac{2}{3} =  \frac{1}{3}$, where the first inequality follows from the subadditivity.
    Thus, it is clear that its Nash welfare is at least $\frac{\sqrt{2}}{3}$-max-NSW.

    \item If $v_1(e^*) < v_2(M \setminus \{e^*\})$, it is easy to find a subset $T \subseteq M \setminus \{e^*\}$ that satisfies $v_2(T) \geq v_1(e^*)$ and $v_2(T \setminus \{e\}) < v_1(e^*)$ for some $e \in T$, and then $T$ is allocated to agent 2.
    Then, this partial allocation is EQ1.
    In addition, we get $v_1(e^*) \geq \frac{2}{3}$ and $v_2(T) \geq v_1(e^*) \geq \frac{2}{3}$, so we can conclude that its Nash welfare is at least $\frac{2}{3}$-max-NSW.
\end{itemize}
    
\textbf{Case 3.} If there exists an item $e^*$ such that $v_1(e^*) \geq \frac{2}{3}$ and $v_2(e^*) \geq \frac{2}{3}$, item $e^*$ is allocated to the agent who has the highest valuation and the remaining items are allocated to the other agent.
It is easy to see that this partial allocation is EQ1.
Next, we bound the Nash welfare.
Without loss of generality, we assume that agent 1 has the highest valuation, i.e., $v_1(e^*) \geq v_2(e^*)$.
In that case, we have $v_1(M\setminus \{e^*\}) \leq v_2(M\setminus \{e^*\})$.
Let  $(A_1^*, A_2^*)$ denote the maximum Nash welfare allocation.
If $e^* \in A_1^*$, we have $v_1(e^*) \geq \frac{2}{3}v_1(A_1^*)$, where the inequality holds since $v_1(A_1^*) \leq 1$, and $v_2(M \setminus \{e^*\}) \geq v_2(A_2^*)$, where the inequality follows from the fact that $A_2^* \subseteq M \setminus \{e^*\}$.
If not, we have $v_1(e^*) \geq \frac{2}{3}v_2(A_2^*)$ since $v_2(A_2^*) \leq 1$, and $v_2(M \setminus \{e^*\}) \geq v_1(A_1^*)$ since $A_1^* \subseteq M \setminus \{e^*\}$ and $v_1(M\setminus \{e^*\}) \leq v_2(M\setminus \{e^*\})$.
Thus, the Nash welfare of this partial allocation is at least $\sqrt{\frac{2}{3}}$-max-NSW.
    

\textbf{Case 4.} If for any item $e \in M$, we have $v_i(e) < \frac{2}{3}$ for any $i \in [2]$. Then, we consider two subcases.
\begin{itemize}
    \item If there exists an item $e^* \in M$ such that $v_i(e^*) \geq \frac{1}{3}$ for some $i \in [2]$,
    without loss of generality, we assume that agent 1 is such an agent and item $e^*$ is allocated to agent 1.
    If $v_1(e^*) \geq v_2(M \setminus \{e^*\})$, $M \setminus \{e^*\}$ is allocated to agent 2.
    It is not hard to see that this partial allocation is EQ1.
    In addition, we have $v_2(M \setminus \{e^*\}) \geq v_2(M) -v_2(e^*) > 1-\frac{2}{3}= \frac{1}{3}$, where the inequality follows for the subadditivity.
    If $v_1(e^*) < v_2(M \setminus \{e^*\})$, it is easy to find $T \subseteq M \setminus \{e^*\}$ such that $v_2(T) \geq v_1(e^*)$ and $v_2(T \setminus \{e\}) < v_1(e^{*})$ for some $e \in T $, and then $T$ is allocated to agent 2.
    It is trivial that this partial allocation is EQ1 and $v_2(T)\geq v_1(e^*) \geq  \frac{1}{3}$.
   At last, in the above two cases, the Nash welfare is at least $\frac{1}{3}$-max-NSW.
    \item If for any item $e \in M$, we have $v_i(e) < \frac{1}{3}$ for any $i \in [2]$.
    Then, the algorithm adopts the greedy way to find two bundles, $A_1$ and $A_2$.
    It is clear that if the two bundles are successfully found, this partial allocation is EQ1, and its Nash welfare is at least $\frac{1}{3}$-max-NSW.
    It suffices to show that after one agent receives bundle $T$, the value of the remaining items for the other agent is at least $\frac{1}{3}$.
    Suppose that agent 2 is the second agent.
    When agent 1 picks $T$, for agent 2, we have $v_2(T)  \leq  v_2(T \setminus \{e_{\ell}\} ) + v_2(e_{\ell}) <  \frac{1}{3} +\frac{1}{3} = \frac{2}{3}$ where the inequality follows from the subadditivity.
    Therefore, we get $v_2(M \setminus T) \geq v_2(M) -v_2(T) \geq \frac{1}{3} $ where the inequality also follows from the subadditivity.
\end{itemize}
Combining the above cases together establishes the correctness of the algorithm.
\end{proof}

We further study the normalized additive valuation setting for Nash welfare with two agents, presenting an algorithm to compute an EQ1 allocation with a better approximation of maximum Nash welfare compared to the subadditive valuation setting.
The key distinction between this algorithm and the previous one lies in our approach: we begin with a maximum Nash welfare allocation and use the property that for additive valuations, such an allocation is always envy-free up to one item (EF1) \cite{caragiannis2019unreasonable}, which does not hold with the subadditive valuations. 
Similar to Algorithm \ref{alg: eq1_nsw_computation}, our approach requires a detailed case analysis, while maintaining the underlying intuition remains consistent: we first allocate the highest-valued item, followed by the remaining items.
The complete algorithm details are provided in Algorithm \ref{alg: eq1_nsw_existence}.

It is worth noting that for additive valuations, given any instance, there exists an EF1 allocation that maximizes Nash welfare.
Surprisingly, based on the instance constructed in the proof of Lemma \ref{the: subsidy_no_efficiency_scaled}, any EQ1 allocation can achieve an arbitrarily poor approximation of maximum Nash welfare, which is in sharp contrast to EF1.

\begin{theorem}
\label{the: additive_two_agents_nash_welfare_existence}
    For two agents with normalized additive valuations in the goods setting, Algorithm \ref{alg: eq1_nsw_existence} computes an EQ1 allocation that has $\frac{1}{\sqrt{2}}$-max-NSW.
\end{theorem}
\begin{proof}
 The algorithm starts with a maximum Nash welfare allocation. 
 If it is EQ1, we are done. 
 If not, we consider several cases to compute a partial EQ1 allocation that guarantees $\frac{1}{\sqrt{2}}$-max-NSW.
 Then, this partial allocation can be extended to a complete EQ1 allocation by Algorithm \ref{alg: nonmonotone_subsidy}, and the Nash welfare does not decrease in this phase.
 Without loss of generality, we assume that $v_1(A_1^*) > v_2(A_2^*)$ and this allocation is not EQ1.

 \textbf{Case 1.} If $v_2(A_2^*) \geq \frac{1}{2}v_1(A_1^*)$, it is not hard to find a subset $T \subseteq A_1^*$ such that $v_1(T) \geq v_2(A_2^*)$ and $v_1(T \setminus \{e\}) < v_2(A_2^*)$ for some $e \in T$.
 Then, $T$ is allocated to agent 1, and $A_2^*$ is allocated to agent 2.
 This partial allocation is trivially EQ1.
 Since $v_1(T) \geq v_2(A_2^*) \geq \frac{1}{2}v_1(A_1^*)$ and $A_2^*$ keeps the same, this partial allocation has $\frac{1}{\sqrt{2}}$-max-NSW.

\textbf{Case 2.} If $v_2(A_2^*) < \frac{1}{2}v_1(A_1^*)$, we claim that $v_2(A_2^*) \leq \frac{1}{2} \leq v_1(A_1^*)$.
If $v_2(A_2^*) > \frac{1}{2}$, we get $v_1(A_1^*) > 1$, which is the contradiction.
If $v_1(A_1^*) < \frac{1}{2}$, $A_1^*$ is allocated to agent 2 and $A_2^*$ is allocated to agent 1, where this new allocation has a strictly better Nash welfare and contradicts that ${\bf A}^*$ is a maximum Nash welfare allocation.
Next, let us consider the following subcases.

\begin{itemize}
    \item If $v_1(A_1^*) \leq \frac{2}{3}$, we have $v_1(A_2^*) \geq \frac{1}{2}v_1(A_1^*)$.
    Then, it is easy to find a subset $T \subseteq A_1^*$ such that $v_2(T) \geq v_1(A_2^*)$ and $v_2(T \setminus \{e\}) < v_1(A_2^*)$ for some $e \in T$.
    Then, $T$ is allocated to agent 2, and $A_2^*$ is allocated to agent 1.
    It is clear that this partial allocation is EQ1.
    Since $v_2(T) \geq v_1(A_2^*) \geq \frac{1}{2}v_1(A_1^*) \geq v_2(A_2^*)$, this partial allocation has $\frac{1}{\sqrt{2}}$-max-NSW.
    \item If there exists an item $e^* \in A_1^*$ such that $v_1(e^*) \geq \frac{1}{2}v_1(A_1^*)$, item $e^*$ is allocated to agent 1.
    Consider two different cases about $v_2(M \setminus \{e^*\})$.
    If $v_2(M \setminus \{e^*\}) < v_1(e^*)$, $M \setminus \{e^*\}$ is allocated to agent 2.
    This allocation is EQ1 and has a $\frac{1}{\sqrt{2}}$-max-NSW.
    If $v_2(M \setminus \{e^*\}) \geq v_1(e^*)$, it is easy to find a subset $T \subseteq M \setminus \{e^*\}$ such that $v_2(T) \geq v_1(e^*)$ and $v_2(T \setminus \{e\}) < v_1(e^*)$ for some $e \in T$.
    Then, $T$ is allocated to agent 2, and this partial allocation is EQ1.
    Since $v_2(T) \geq v_1(e^*) \geq \frac{1}{2}v_1(A_1^*) \geq v_2(A_2^*)$, this partial allocation has $\frac{1}{\sqrt{2}}$-max-NSW.
    \item If for any item $e \in A_1^*$, we have $v_1(e) < \frac{1}{2}v_1(A_1^*)$.
    Then, let $e_{max}$ denote the largest value in $A_1^*$ to agent 2.
    Now, we get $v_1(A_1^* \setminus \{e_{max}\}) \geq \frac{1}{2}v_1(A_1^*)$. 
    Since ${\bf A}^*$ is EF1, we have $v_2(A_2^* \cup \{e_{max}\}) \geq \frac{1}{2} \geq \frac{1}{2}v_1(A_1^*)$.
    If $ v_1(A_1^* \setminus \{e_{max}\}) < v_2(A_2^* \cup \{e_{max}\})$,
    it is easy to find a subset $T \subseteq A_2^* \cup \{e_{max}\}$ such that $v_2(T) \geq v_1(A_1^* \setminus \{e_{max}\})$ and $v_2(T \setminus \{e\}) < v_1(A_1^* \setminus \{e_{max}\})$ for some $e \in T$.
    Then,  $A_1^* \setminus \{e_{max}\}$ is allocated to agent 1 and $T$ is allocated to agent 2.
    Since $v_2(T) \geq v_1(A_1^* \setminus \{e_{max}\}) \geq \frac{1}{2}v_1(A_1^*) \geq v_2(A_2^*)$, this partial allocation has $\frac{1}{\sqrt{2}}$-max-NSW.  
    If $ v_1(A_1^* \setminus \{e_{max}\}) \geq v_2(A_2^* \cup \{e_{max}\})$,
    it is easy to find a subset $T \subseteq A_1^* \setminus \{e_{max}\}$ such that $v_1(T) \geq v_2(A_2^* \cup \{e_{max}\})$ and $v_1(T \setminus \{e\}) < v_2(A_2^* \cup \{e_{max}\})$ for some $e \in T$.
    Then,  $T$ is allocated to agent 1 and $A_2^* \cup \{e_{max}\}$ is allocated to agent 2.
    This partial allocation is trivially EQ1, and it has $\frac{1}{\sqrt{2}}$-max-NSW.  
\end{itemize}
Combining the above cases together establishes the correctness of the algorithm.
 \end{proof}

Finally, we give the upper bound of the fraction of the maximum Nash welfare that an EQ1 allocation guarantees for two agents with normalized additive valuations to complement the above results.

\begin{proposition}
\label{pro: upper_bound_nash_two_agents}
    In the goods setting, there exists an instance where no EQ1 allocation has  $\frac{\sqrt{3}}{2}$-max-NSW for two agents with normalized additive valuations.
\end{proposition}

\begin{proof}
 Consider the following instance with two agents and three items shown in Table \ref{tab: eq1_nsw_two_agents}. 
 It is trivial that in the maximum Nash welfare allocation, agent 1 picks items $e_1$ and $e_2$, and agent 2 picks item $e_3$.
 The maximum Nash welfare is $\sqrt{\frac{1}{3}+2\epsilon}$.
 However, this allocation is not EQ1.
 Thus, in any EQ1 allocation, agent 1 cannot receive both items $e_1$ and $e_2$.
 Consider an EQ1 allocation ${\bf A}$, item $e_1$ is allocated to agent 1, and the remaining items are allocated to agent 2.
The Nash welfare of this allocation is $\sqrt{\frac{1}{3}+\frac{1}{2}\epsilon}$.
 Therefore, the worst ratio between Nash welfare that EQ1 allocations has and maximum Nash welfare is $\lim_{\epsilon \rightarrow \frac{1}{12}}\sqrt{\frac{\frac{1}{3}+\frac{1}{2}\epsilon}{\frac{1}{3}+2\epsilon}} = \frac{\sqrt{3}}{2}$.
 \begin{table}[tb]
    \centering
    \begin{tabular}{c|c|c|c}
    \toprule
      & $e_1$ & $e_2$ & $e_3$\\
      \midrule
        agent 1  & $\frac{1}{2}$ & $\frac{1}{2}$ & 0 \\
        \midrule
        agent 2 & $\frac{1}{3}-\epsilon$ & $\frac{1}{3}-\epsilon$&  $\frac{1}{3}+2\epsilon$\\
        \bottomrule
    \end{tabular}   \caption{Example showing that the upper bound of approximate Nash welfare that an EQ1 allocation has, where $0<\epsilon<\frac{1}{12}$.}
    \label{tab: eq1_nsw_two_agents}
\end{table}
 \end{proof}

 \begin{algorithm}[!]
\caption{Two agents with subadditive valuations}
\label{alg: eq1_nsw_computation}
\KwIn{A goods instance $\mathcal{I} = \langle M, 2, \boldsymbol{v} \rangle$ with normalized subadditive valuations}
\KwOut{An EQ1 allocation $\mathcal{{\bf A}}$ with the ``optimally bounded'' subsidy and $\frac{1}{3}$-max-NSW }

Let ${\bf A} = (\emptyset, \emptyset)$;

\uIf{$\exists  e_1, e_2 \in M $ such that $v_1(e_1) \geq \frac{2}{3}$ and $v_2(e_2) \geq \frac{2}{3}$}{

$A_1 = \{e_1\}$ and $A_2 = \{e_2\}$;

}
\uElseIf{$\exists e^* \in M$, $i \in [2]$ such that $v_i(e^*) \geq \frac{2}{3}$ }{
Assume that $v_1(e^*) \geq \frac{2}{3}$;

\eIf{$v_1(e^*) \geq v_2(M \setminus \{e^*\})$}{
$A_1 = \{e^*\}$ and $A_2 = M \setminus \{e^*\}$;}
{
Find a subset $T \subseteq M \setminus \{e^*\}$ such that $v_2(T) \geq v_1(e^*)$ and $v_2(T \setminus \{e\}) \leq v_1(e^*)$ for some $e \in t$;

$A_1 = \{e^*\}$ and $A_2 = T$;
}
}
\uElseIf{$\exists e^* \in M$ such that $v_i(e^*) \geq \frac{2}{3}$ for any $i \in [2]$}{

\eIf{$v_1(e^*) \geq v_2(e^*)$}{

$A_1 = \{e^*\}$ and $A_2 = M \setminus  \{e^*\}$;

}
{\{$A_1 = M \setminus  \{e^*\}$  and $A_2 =  \{e^*\}$;}

}

\Else{

\eIf{$\exists e^{*} \in M$ such that $v_i(e^{*}) \geq \frac{1}{3}$ for some $i \in [2]$ }{

Assume that $v_1(e^*) \geq \frac{1}{3}$;

\eIf{$v_1(e^*) \geq v_2(M \setminus \{e^*\})$}{
$A_1 = \{e^*\}$ and $A_2 = M \setminus \{e^*\}$;
}
{
Find a subset $T \subseteq M \setminus \{e^{*}\}$ such that $v_2(T) \geq v_1(e^{*})$ and $v_2(T \setminus \{e\}) < v_1(e^{*})$ for some $e \in T$;

$A_1 = \{e^{*}\}$ and $A_2 = T$;
}

}
{

Let $N^{-} = N$ and $M^{\prime} = M$;

\While{$N^{-}\neq \emptyset$}{
Let $T = \emptyset$;

\While{$\forall i \in N^{-}$, $v_i(T) < \frac{1}{3}$}{
Pick an arbitrary item $e \in M^{\prime}$;

$T = T \cup \{e\}$ and $M^{'} = M^{'} \setminus \{e\}$; 

}

Let $N^{*} = \{i \in N^{-}| v_i(T) \geq \frac{1}{3} \}$ and choose an arbitrary agent $i \in N^{-}$;

$A_i = T$ and $N^{-} = N^{-} \setminus \{i\}$;
}








}

}

\uIf{$M \setminus (A_1 \cup A_2) \neq \emptyset$}{
Execute Algorithm \ref{alg: nonmonotone_subsidy} to extend the partial allocation ${\bf A}$ to a complete one;
}

\Return ${\bf A}$;

\end{algorithm}


\begin{algorithm}[!]
\caption{Two agents with additive valuations}
\label{alg: eq1_nsw_existence}
\KwIn{A goods instance $\mathcal{I} = \langle M, 2, \boldsymbol{v} \rangle$ with normalized additive valuations, and a maximum Nash welfare allocation ${\bf A}^* = (A_1^{*}, A_{2}^{*})$.}
\KwOut{An EQ1 allocation $\mathcal{{\bf A}}$ with the ``optimally bounded'' subsidy and $ \frac{1}{\sqrt{2}}$-max-NSW}

Let ${\bf A} = (\emptyset, \emptyset)$;

\eIf{${\bf A}^{*}$ is not EQ1}{
 Assume that $v_i(A_1^{*}) \geq v_2(A_2^{*})$;

\eIf{$v_2(A_2^*)\geq \frac{1}{2}v_1(A_1^{*})$}{

Find a subset $T\subseteq A_1^*$ such that $v_1(T) \geq v_2(A_2^*)$ and $v_1(T \setminus \{e\}) \leq v_2(A_2^*)$ for some $e \in T$; 

$A_1 = T$ and $A_2 = A_2^{*}$;

}
{
\eIf{$ \frac{1}{2} \leq v_1(A_1^{*}) \leq \frac{2}{3}$}{
Find a subset $T \subseteq A_1^*$ such that $v_2(T) \geq v_1(A_2^*)$ and $v_2(T \setminus \{e\}) < v_1(A_2^*)$ for some $e \in T$;

$A_1 = A_2^*$ and $A_2 = T$;
}
{
\eIf{$\exists e^* \in A_1^*$ such that $v_1(e^*) \geq \frac{1}{2}v_1(A_1^*)$}{
\eIf{$v_1(e^*) \geq  v_2(M\setminus \{e^*\})$}{

$A_1 = \{e^*\}$ and $A_2 = M \setminus \{e^*\}$;
}
{

Find a subset $T \subseteq M \setminus \{e^*\}$ such that $v_2(T) \geq v_1(e^*)$ and $v_2(T \setminus \{e\}) < v_1(e^*)$ for some $e \in T$;

$A_1 = \{e^*\}$ and $A_2 = T$;

}

}
{
Let $e_{max} \in \arg \max_{e \in A_1^*}v_2(e)$;

\eIf{$v_1(A_1^* \setminus \{e_{max}\}) \geq   v_2(A_2^* \cup   \{e_{max}\})$}{

Find a subset $T \subseteq A_2^* \cup \{e_{max}\}$ such that $v_2(T) \geq v_1(A_1^* \setminus \{e_{max}\})$ and $v_2(T \setminus \{e\}) < v_1(A_1^* \setminus \{e_{max}\})$ for some $e \in T$;

$A_1 = A_1^* \setminus \{e_{max}\}$ and $A_2 = T$;
}
{
Find a subset $T \subseteq A_1^* \setminus \{e_{max}\}$ such that $v_1(T) \geq v_2(A_2^* \cup \{e_{max}\})$ and $v_1(T \setminus \{e\}) < v_2(A_2^* \cup \{e_{max}\})$ for some $e \in T$;

$A_1 = T$ and $A_2 = A_2^* \cup \{e_{max}\}$;

}
}

}
}

}
{$A_1 = A_1^*$ and $A_2 = A_2^*$;}

\uIf{ $M \setminus (A_1 \cup A_2) \neq \emptyset$}{


Execute Algorithm \ref{alg: nonmonotone_subsidy} to extend the partial allocation ${\bf A}$ to a complete one;
}

\Return ${\bf A}$;

\end{algorithm}

\subsection{Equitability and Envy-Freeness with Efficiency Guarantees}
This part examines the combination of two notions of fairness, equitability and envy-freeness.
From the previous section, we establish that when one allocation satisfies the conditions specified in Lemma \ref{the: characterization_ef_eq}, the required minimum subsidy for achieving both equitability and envy-freeness is ``optimally bounded''.
Thus, our focus shifts to identifying allocations that are both equitable-convertible and envy-freeable while simultaneously guaranteeing some level of efficiency.
We begin by presenting a negative result demonstrating the incompatibility between the "equitable-convertible and envy-freeable" property and Nash welfare.
\begin{proposition}
\label{prop: no_efficiency_ef_eq}
    For any $\alpha>0$, in the goods setting, there exists an instance where no allocation that is both equitable-convertible and envy-freeable can guarantee $\alpha$-max NSW, even for additive valuations.
\end{proposition}
  \begin{proof}
Consider an instance with $n$ agents who have additive valuations and $m$ items. 
For agent $1$, the value of each item is $1$.
For the remaining agent $i \neq 1$, the value of each item is $(1-\epsilon)$ $(0<\epsilon \leq 1)$.
By Theorem \ref{the: characterization_ef_eq}, in this instance, the only equitable-convertible and envy-freeable allocation ${\bf A}$ is that agent $1$ picks all items since for any item $e \in M$, we have $v_1(e) > v_i(e)$ for any $i \neq 1$. 
Then, we can derive that the Nash welfare of allocation $\bf{A}$ is 0, but the maximum Nash welfare of this instance is trivially non-zero.
\end{proof}

However, for additive valuations, we show that a maximum utilitarian social welfare allocation is always equitable-convertible and envy-freeable.
\begin{proposition}
\label{the: additive_utilitarian_eq_ef}
    For additive valuations in the goods setting, maximum utilitarian social welfare allocations are both equitable-convertible and envy-freeable. 
\end{proposition}

\begin{proof}
 Given any instance, consider a maximum utilitarian social welfare allocation ${\bf A}^{*}$. 
Now pick an arbitrary agent $i \in N$.
For any item $e \in A_{i}^{*}$, it must hold that $v_i(e) \geq v_j(e)$ for any agent $j \neq i$. 
 Otherwise, we could find another allocation that has a higher utilitarian social welfare.
 Thus, we get $v_i(A_{i}^{*})  = \sum_{e \in A_{i}^{*}}v_i(e) \geq \sum_{e \in A_{i}^{*}}v_j(e) = v_j(A_{i}^{*})$ for any agent $j \neq i$.
 By Theorem \ref{the: characterization_ef_eq}, it is clear that ${\bf A}^{*}$ is both equitable-convertible and envy-freeable.
 \end{proof}

Then, we complement the above result by showing that if agents have beyond additive valuations like submodular valuations, the above theorem may not hold.

\begin{proposition}
\label{pro: submodular_maximizing_social_welfare}
    For submodular valuations in the goods setting, there exists an instance where maximum utilitarian social welfare allocations are not both equitable-convertible and envy-freeable. 
\end{proposition}

\begin{proof}
    Consider an instance with two agents who have submodular valuations and four items $M = \{e_1, e_2, e_3, e_4\}$ that are all goods.
    The valuation of each agent is given as follows.
    For any subset $S \subseteq M$, we have
    $$v_1(S)=\left\{
    \begin{aligned}
        1& \quad \text{if} ~ |S|=1,\\
        \frac{5}{4}& \quad \text{if}~ |S|\ge 2
    \end{aligned}
    \right.
    $$.
    $$v_2(S)=\left \{
    \begin{aligned}
        1& \quad \text{if}~ |S|=1,\\
        \frac{3}{2}& \quad \text{if}~ |S|\ge 2
    \end{aligned}
    \right.$$
 It can be verified that in any maximum utilitarian social welfare allocation, each agent only gets two items. 
 Let ${\bf A}=(A_1, A_2)$ be a maximum utilitarian social welfare allocation, where $A_1=\{e_1, e_2\}$ and $A_2=\{e_3, e_4\}$.
 Then, we have $v_1(A_1) < v_2(A_1)$, so allocation ${\bf A}$ cannot be both equitable-convertible and envy-freeable.
\end{proof}

We now turn our attention to a special case of submodular valuation functions known as matroid rank valuations, which have been extensively studied in the fair division literature \cite{barman2021existence,barman2022achieving}.
Remarkably, for matroid rank valuations, 
we show that allocations either maximizing utilitarian social welfare or Nash welfare are both equitable-convertible and envy-freeable, and can be computed efficiently.



\begin{theorem}
\label{the: matroid_rank_eq_ef}
For matroid rank valuations in the goods setting, there exists a maximum Nash welfare (or maximum utilitarian social welfare) allocation that is both equitable-convertible and envy-freeable, and both of them can be computed in polynomial time. 
\end{theorem}

Before we show Theorem \ref{the: matroid_rank_eq_ef}, we give the definition of clean allocations and a key lemma.

\begin{definition}[Clean Allocation]
An allocation ${\bf A}$ is said to be clean if $v_i(A_i)>v_i(A_i\setminus \{e\})$ for any agent $i\in N$ with $A_i \neq \emptyset$ and any item $e\in A_i$    
\end{definition}

\begin{lemma}\cite{benabbou2021finding}
\label{lem: matroid_rank_utilitarian_optimal_optimL}
   For matroid rank valuations in the goods setting, a clean maximum utilitarian social welfare allocation can be computed in polynomial time.
\end{lemma}

\begin{proof}[Proof of Theorem \ref{the: matroid_rank_eq_ef}]
 First, we consider the maximum Nash welfare allocation. 
 Let ${\bf A}=(A_1, \ldots, A_n)$ be the allocation computed by the General Yankee Swap Algorithm with the parameter $w_i = \frac{1}{n}$ for any $i \in N$ \cite{Viswanathan23}. 
 It is not hard to check that ${\bf A}$ is a clean maximum Nash welfare allocation.
 For each agent $i\in N$, we have $v_i(A_i)>v_i(A_i\setminus\{e\})$ for any $e\in A_i$. Since $v_i$ is a matroid rank valuation, we get $v_i(A_i)=|A_i|$, then, $v_i(A_i)\ge v_j(A_i)$ for all $j\ne i$. 
 By Theorem \ref{the: characterization_ef_eq}, it is clear that allocation ${\bf A}$ is both equitable-convertible and envy-freeable.
 
 Next, consider the maximum utilitarian social welfare allocation.
 By Lemma \ref{lem: matroid_rank_utilitarian_optimal_optimL}, it is clear that a clean maximum utilitarian social welfare allocation can be computed in polynomial time.
 Then, the proof is similar to the above one, and we omit it.
\end{proof}

\section{Conclusion and Future Work}
In this paper, we investigate the fairness notion of equitability and study the subsidy for achieving equitability, and present a series of positive and negative results.
For future work, one promising direction involves bridging the gap between the lower and upper bounds of the required minimum subsidy for achieving equitability in the mixed items setting with additive or more general valuations. 
Another interesting direction is to study what Nash welfare guarantees can be achieved for EQ1 allocations when there are more than two agents with normalized subadditive or additive valuations.

\newpage
\bibliographystyle{plainnat}
\bibliography{ref}

\begin{thebibliography}{29}
\providecommand{\natexlab}[1]{#1}
\providecommand{\url}[1]{\texttt{#1}}
\expandafter\ifx\csname urlstyle\endcsname\relax
  \providecommand{\doi}[1]{doi: #1}\else
  \providecommand{\doi}{doi: \begingroup \urlstyle{rm}\Url}\fi

\bibitem[Alkan et~al.(1991)Alkan, Demange, and Gale]{alkan1991fair}
Ahmet Alkan, Gabrielle Demange, and David Gale.
\newblock Fair allocation of indivisible goods and criteria of justice.
\newblock \emph{Econometrica: Journal of the Econometric Society}, pages 1023--1039, 1991.

\bibitem[Aziz(2021)]{aziz2021achieving}
Haris Aziz.
\newblock Achieving envy-freeness and equitability with monetary transfers.
\newblock In \emph{Proceedings of the 35th AAAI Conference on Artificial Intelligence}, pages 5102--5109, 2021.

\bibitem[Barman et~al.(2022)Barman, Krishna, Narahari, and Sadhukan]{barman2022achieving}
S~Barman, A~Krishna, Y~Narahari, and S~Sadhukan.
\newblock Achieving envy-freeness with limited subsidies under dichotomous valuations.
\newblock In \emph{Proceedings of the 31st International Joint Conference on Artificial Intelligence}, pages 60--66, 2022.

\bibitem[Barman and Verma(2021)]{barman2021existence}
Siddharth Barman and Paritosh Verma.
\newblock Existence and computation of maximin fair allocations under matroid-rank valuations.
\newblock In \emph{Proceedings of the 20th International Conference on Autonomous Agents and MultiAgent Systems}, pages 169--177, 2021.

\bibitem[Barman et~al.(2020)Barman, Bhaskar, and Shah]{barman2020optimal}
Siddharth Barman, Umang Bhaskar, and Nisarg Shah.
\newblock Optimal bounds on the price of fairness for indivisible goods.
\newblock In \emph{International Conference on Web and Internet Economics}, pages 356--369, 2020.

\bibitem[Barman et~al.(2024)Barman, Bhaskar, Pandit, and Pyne]{barman2024nearly}
Siddharth Barman, Umang Bhaskar, Yeshwant Pandit, and Soumyajit Pyne.
\newblock Nearly equitable allocations beyond additivity and monotonicity.
\newblock In \emph{Proceedings of the 38th AAAI Conference on Artificial Intelligence}, pages 9494--9501, 2024.

\bibitem[Benabbou et~al.(2021)Benabbou, Chakraborty, Igarashi, and Zick]{benabbou2021finding}
Nawal Benabbou, Mithun Chakraborty, Ayumi Igarashi, and Yair Zick.
\newblock Finding fair and efficient allocations for matroid rank valuations.
\newblock \emph{ACM Transactions on Economics and Computation}, 9\penalty0 (4):\penalty0 1--41, 2021.

\bibitem[Bouveret and Lang(2008)]{bouveret2008efficiency}
Sylvain Bouveret and J{\'e}r{\^o}me Lang.
\newblock Efficiency and envy-freeness in fair division of indivisible goods: Logical representation and complexity.
\newblock \emph{Journal of Artificial Intelligence Research}, 32:\penalty0 525--564, 2008.

\bibitem[Brustle et~al.(2020)Brustle, Dippel, Narayan, Suzuki, and Vetta]{brustle2020one}
Johannes Brustle, Jack Dippel, Vishnu~V Narayan, Mashbat Suzuki, and Adrian Vetta.
\newblock One dollar each eliminates envy.
\newblock In \emph{Proceedings of the 21st ACM Conference on Economics and Computation}, pages 23--39, 2020.

\bibitem[Bu et~al.(2025)Bu, Li, Liu, Song, and Tao]{bu2025approximability}
Xiaolin Bu, Zihao Li, Shengxin Liu, Jiaxin Song, and Biaoshuai Tao.
\newblock Approximability landscape of welfare maximization within fair allocations.
\newblock In \emph{Proceedings of the 26th ACM Conference on Economics and Computation}, pages 412--440, 2025.

\bibitem[Caragiannis and Ioannidis(2021)]{caragiannis2021computing}
Ioannis Caragiannis and Stavros~D Ioannidis.
\newblock Computing envy-freeable allocations with limited subsidies.
\newblock In \emph{International Conference on Web and Internet Economics}, pages 522--539, 2021.

\bibitem[Caragiannis et~al.(2019)Caragiannis, Kurokawa, Moulin, Procaccia, Shah, and Wang]{caragiannis2019unreasonable}
Ioannis Caragiannis, David Kurokawa, Herv{\'e} Moulin, Ariel~D Procaccia, Nisarg Shah, and Junxing Wang.
\newblock The unreasonable fairness of maximum nash welfare.
\newblock \emph{ACM Transactions on Economics and Computation}, 7\penalty0 (3):\penalty0 1--32, 2019.

\bibitem[Caragiannis et~al.(2021)Caragiannis, Kanellopoulos, and Kyropoulou]{caragiannis2021interim}
Ioannis Caragiannis, Panagiotis Kanellopoulos, and Maria Kyropoulou.
\newblock On interim envy-free allocation lotteries.
\newblock In \emph{Proceedings of the 22nd ACM Conference on Economics and Computation}, pages 264--284, 2021.

\bibitem[Freeman et~al.(2019)Freeman, Sikdar, Vaish, and Xia]{freeman2019equitable}
Rupert Freeman, Sujoy Sikdar, Rohit Vaish, and Lirong Xia.
\newblock Equitable allocations of indivisible goods.
\newblock In \emph{Proceedings of the 28th International Joint Conference on Artificial Intelligence}, pages 280--286, 2019.

\bibitem[Freeman et~al.(2020)Freeman, Sikdar, Vaish, and Xia]{freeman2020equitable}
Rupert Freeman, Sujoy Sikdar, Rohit Vaish, and Lirong Xia.
\newblock Equitable allocations of indivisible chores.
\newblock In \emph{Proceedings of the 19th International Conference on Autonomous Agents and MultiAgent Systems}, pages 384--392, 2020.

\bibitem[Goko et~al.(2024)Goko, Igarashi, Kawase, Makino, Sumita, Tamura, Yokoi, and Yokoo]{goko2024fair}
Hiromichi Goko, Ayumi Igarashi, Yasushi Kawase, Kazuhisa Makino, Hanna Sumita, Akihisa Tamura, Yu~Yokoi, and Makoto Yokoo.
\newblock A fair and truthful mechanism with limited subsidy.
\newblock \emph{Games and Economic Behavior}, 144:\penalty0 49--70, 2024.

\bibitem[Goldman and Procaccia(2015)]{goldman2015spliddit}
Jonathan Goldman and Ariel~D Procaccia.
\newblock Spliddit: Unleashing fair division algorithms.
\newblock \emph{ACM SIGecom Exchanges}, 13\penalty0 (2):\penalty0 41--46, 2015.

\bibitem[Gourv{\`e}s et~al.(2014)Gourv{\`e}s, Monnot, and Tlilane]{gourves2014near}
Laurent Gourv{\`e}s, J{\'e}r{\^o}me Monnot, and Lydia Tlilane.
\newblock Near fairness in matroids.
\newblock In \emph{21st European Conference on Artificial Intelligence}, pages 393--398, 2014.

\bibitem[Halpern and Shah(2019)]{halpern2019fair}
Daniel Halpern and Nisarg Shah.
\newblock Fair division with subsidy.
\newblock In \emph{The 12th International Symposium on Algorithmic Game Theory}, pages 374--389, 2019.

\bibitem[Herreiner and Puppe(2009)]{herreiner2009envy}
Dorothea~K Herreiner and Clemens~D Puppe.
\newblock Envy freeness in experimental fair division problems.
\newblock \emph{Theory and decision}, 67:\penalty0 65--100, 2009.

\bibitem[Hosseini and Sethia(2025)]{hosseini2025equitable}
Hadi Hosseini and Aditi Sethia.
\newblock Equitable allocations of mixtures of goods and chores.
\newblock \emph{arXiv preprint arXiv:2501.06799}, 2025.

\bibitem[Kawase et~al.(2024)Kawase, Makino, Sumita, Tamura, and Yokoo]{kawase2024towards}
Yasushi Kawase, Kazuhisa Makino, Hanna Sumita, Akihisa Tamura, and Makoto Yokoo.
\newblock Towards optimal subsidy bounds for envy-freeable allocations.
\newblock In \emph{Proceedings of the 38th AAAI Conference on Artificial Intelligence}, pages 9824--9831, 2024.

\bibitem[Maskin(1987)]{maskin1987fair}
Eric~S Maskin.
\newblock On the fair allocation of indivisible goods.
\newblock In \emph{Arrow and the Foundations of the Theory of Economic Policy}, pages 341--349. 1987.

\bibitem[Narayan et~al.(2021)Narayan, Suzuki, and Vetta]{narayan2021two}
Vishnu~V Narayan, Mashbat Suzuki, and Adrian Vetta.
\newblock Two birds with one stone: Fairness and welfare via transfers.
\newblock In \emph{The 14th International Symposium on Algorithmic Game Theory}, pages 376--390, 2021.

\bibitem[Segal-Halevi and Sziklai(2018)]{segal2018resource}
Erel Segal-Halevi and Bal{\'a}zs~R Sziklai.
\newblock Resource-monotonicity and population-monotonicity in connected cake-cutting.
\newblock \emph{Mathematical Social Sciences}, 95:\penalty0 19--30, 2018.

\bibitem[Sun et~al.(2023)Sun, Chen, and Doan]{sun2023equitability}
Ankang Sun, Bo~Chen, and Xuan~Vinh Doan.
\newblock Equitability and welfare maximization for allocating indivisible items.
\newblock \emph{Autonomous Agents and Multi-Agent Systems}, 37\penalty0 (1):\penalty0 8, 2023.

\bibitem[Viswanathan and Zick(2023)]{Viswanathan23}
Vignesh Viswanathan and Yair Zick.
\newblock A general framework for fair allocation under matroid rank valuations.
\newblock In \emph{Proceedings of the 24th ACM Conference on Economics and Computation}, page 1129–1152, 2023.

\bibitem[Wu and Zhou(2024)]{wu2024tree}
Xiaowei Wu and Shengwei Zhou.
\newblock Tree splitting based rounding scheme for weighted proportional allocations with subsidy.
\newblock \emph{arXiv preprint arXiv:2404.07707}, 2024.

\bibitem[Wu et~al.(2023)Wu, Zhang, and Zhou]{wu2023one}
Xiaowei Wu, Cong Zhang, and Shengwei Zhou.
\newblock One quarter each (on average) ensures proportionality.
\newblock In \emph{The 19th Conference on Web and Internet Economics}, pages 582--599, 2023.

\end{thebibliography}

\end{document}